\documentclass[11pt,a4paper]{article}
\usepackage[T1]{fontenc}
\usepackage[utf8]{inputenc}
\usepackage{lmodern}
\usepackage{xspace}
\usepackage{authblk}
\usepackage[margin=27mm]{geometry}
\usepackage{amsmath,amssymb,amsthm}
\usepackage{thm-restate}
\usepackage{bbold,bm,float}
\usepackage{needspace}
\usepackage[ruled,vlined]{algorithm2e}
\usepackage[colorlinks=true]{hyperref}
\usepackage[nameinlink,capitalize]{cleveref}
\newtheorem{theorem}{Theorem}[section]

\newtheorem{lemma}[theorem]{Lemma}
\newtheorem{corollary}[theorem]{Corollary}
\newtheorem{proposition}[theorem]{Proposition}
\newtheorem{claim}[theorem]{Claim}
\theoremstyle{definition}
\newtheorem{question}[theorem]{Question}
\theoremstyle{plain}
\crefname{question}{Question}{Questions}
\Crefname{question}{Question}{Questions}
\crefname{claim}{Claim}{Claims}
\Crefname{claim}{Claim}{Claims}
\crefname{figure}{Figure}{Figures}
\Crefname{figure}{Figure}{Figures}
\crefname{equation}{Equation}{Equations}
\Crefname{equation}{Equation}{Equations}
\newenvironment{claimproof}[1]{\begin{proof}[\proofname]}{\end{proof}}
\DeclareMathOperator{\mad}{mad}
\DeclareMathOperator{\arb}{arb}
\DeclareMathOperator{\MDS}{MDS}

\newcommand{\LOCAL}{\textsf{LOCAL}\xspace}
\newcommand{\CONGEST}{\textsf{CONGEST}\xspace}
\xspaceaddexceptions{\textsubscript}
\newcommand{\MinDS}{\textsc{Minimum Dominating Set}\xspace}

\newcommand{\cC}{\mathcal{C}}
\newcommand{\cD}{\mathcal{D}}

\newcommand{\cH}{\mathcal{H}}

\usepackage[dvipsnames]{xcolor}
\definecolor{primary}{HTML}{255077}
\colorlet{secondary}{Dandelion}
\definecolor{tertiary}{HTML}{fe4e63}

\colorlet{mygreen}{YellowGreen}
\colorlet{mydarkgreen}{ForestGreen}
\definecolor{myblack}{HTML}{255077}
\colorlet{myorange}{Dandelion}
\colorlet{myblue}{CornflowerBlue}
\colorlet{mywashedblue}{myblue!60!white}
\colorlet{mypink}{Lavender}
\colorlet{mywashedpink}{mypink!60!white}
\colorlet{mypurple}{violet}
\colorlet{mydarkblue}{RoyalBlue}
\definecolor{myred}{HTML}{fe4e63}

\colorlet{revisionred}{myred!60!black}

\colorlet{rewriteblue}{primary!70!myblue}\colorlet{rewritepurple}{mypurple!45!primary!85!white} 
\DeclareRobustCommand{\rewritewording}[1]{\textcolor{rewriteblue}{#1}}
\DeclareRobustCommand{\unchangedwording}[1]{\textcolor{rewritepurple}{#1}}
\usepackage{tikz}
\usetikzlibrary{math,patterns,positioning,shapes,arrows,arrows.meta,graphs,graphs.standard,decorations.pathmorphing,decorations.markings,fit,tikzmark,backgrounds,calc,hobby}

\tikzset{
    vertex/.style={circle, fill=white, line width=2pt, draw=primary, outer sep=0.2em}
}
\tikzset{
    fvertex/.style={circle, fill=primary!50!black, line width=2pt, draw=white}
}
\tikzset{
    edge/.style={very thick, draw=primary!50!black}
}

\newcommand{\R}{\mathbb{R}}
\newcommand{\indicator}{\mathbb{1}}
\newcommand{\eps}{\varepsilon}
\title{Almost tight distributed approximation for domination in sparse graph classes}

\author[1]{Timoth{\'e} Picavet}
\affil[1]{Gran Sasso Science Institute, L'Aquila, Italy}
\date{}
\hypersetup{
  allcolors=primary,
  pdftitle={Almost tight distributed approximation for domination in sparse graph classes},
  pdfauthor={Timoth{\'e} Picavet}
}

\begin{document}
\maketitle

\begin{abstract}
    We study distributed approximation of \MinDS in a very general class of sparse graphs, extending near-optimal approximation results from the minor-free setting to bounded expansion.  
    In more detail, for fixed $d$ and $\eps>0$, we give a new deterministic constant-round \LOCAL algorithm for \MinDS on graphs of maximum average degree at most $d$ and bounded depth-one shallow-minor density.
    Its approximation ratio is $d+1+\eps$, which is optimal up to $\eps$ for every even integer $d\geq 2$.
    This improves on the recent results of Amiri and Siebertz by extending beyond proper minor-closed classes.
    From this algorithm, we derive almost-tight constant-round approximations for a number of graph classes. 
    We get $(7+\eps)$-approximations for planar graphs and graphs of bounded Euler genus, and $(5+\eps)$-approximations for their triangle-free subclasses.
    This answers the $7+\eps$ question of Wawrzyniak positively, and matches the recent results of Amiri and Siebertz.
    We also apply our algorithm to classes of bounded queue-number or stack-number, obtaining bounds that are tight up to $\eps$.
    Additionally, we answer questions of Amiri and Siebertz about improving the dependence of the round complexity on $1/\eps$ and obtaining polynomial local computation for fixed parameters, when comparing the output to the usual domination number.
\end{abstract}

\section{Introduction}

A \emph{dominating set} of a graph $G$ is a set of vertices whose closed neighborhoods cover $V(G)$.
The \MinDS problem asks for such a set of minimum size, denoted by $\MDS(G)$.
In distributed networks, computing dominating sets is very useful and an important research topic.
For instance, it can be used to organize a network around a small number of active vertices~\cite{KuhnWattenhofer2005}.
In wireless sensor networks, selecting only a few vertices that are active saves energy, by sending nodes outside of the dominating set to sleep~\cite[Chapter~7]{Krishnamachari2005}.
Studying this seemingly basic problem is also a way to quantify what can be achieved with local information.
This is interesting because, even though checking if some vertex is dominated can be easily done locally, globally minimizing the size of the dominating set is a difficult problem.
We focus on constant-round \LOCAL algorithms, where each vertex makes a choice and outputs a return value.
This choice depends only on the topology of the graph and identifiers within a bounded distance, independently of the network size.
For \MinDS, this means that each vertex chooses to be in or outside the returned set, with the guarantee that the returned set dominates the whole graph.
Even on cycles, exactly solving \MinDS requires a nonconstant number of rounds.
Thus, we restrict our attention to approximation algorithms on special graph classes.

\paragraph{Previous results.}
To explain how our result came to be, we first focus on the earlier work that motivated it.
We compare with the recent results of Amiri and Siebertz~\cite{AmiriSiebertz2026} and Wawrzyniak~\cite{Wawrzyniak2026} later on in~\cref{sec:comparison}.

Although lots of special graph classes were studied in the context of distributed domination, such as planar graphs~\cite{CzygrinowEtAl2008,LenzenEtAl2013,Wawrzyniak2014,HilkeEtAl2014,Wawrzyniak2015,HeydtEtAl2022,HeydtEtAl2025,Wawrzyniak2025,Wawrzyniak2026,AmiriSiebertz2026} and outerplanar graphs~\cite{BonamyEtAl2021Outerplanar}, an important step towards a more general result was made by Heydt et al.~\cite[Theorem~2.1]{HeydtEtAl2025}.
They prove that constant-factor approximations in constant rounds can be achieved for graphs of bounded expansion, which is one of the most permissive definitions for sparsity.
While this was a huge step in understanding how \MinDS behaves in sparse graph classes, the proven approximation ratios were huge.
For comparison, bounded-genus graphs were known to admit a $(34+\eps)$-approximation independent of the genus~\cite{BonamyEtAl2026}. Very recently, the ratio was improved to $(25+\eps)$ by Wawrzyniak~\cite[Corollary~7.1]{Wawrzyniak2026}, and then $7+\eps$ by Amiri and Siebertz~\cite[Theorem~8.4]{AmiriSiebertz2026}.
Heydt et al.'s algorithm for $K_{3,t}$-subgraph-free graphs gives a factor $O(1+\sqrt g)$~\cite[Theorem~2.2]{HeydtEtAl2025} for graphs of Euler genus at most $g$.
More generally, if $\nabla_1(G)$ is the largest edge-to-vertex ratio of a simple graph obtained from a subgraph of $G$ by contracting pairwise disjoint stars to vertices, Heydt et al.~\cite[Theorem~2.1]{HeydtEtAl2025} obtain $\nabla_1^{O(\nabla_1^2)}$.
This quantity is more than exponential in $\nabla_1$. 
This leads to the following very natural question: is this bound tight, and if not, how good an approximation can we obtain? Additionally, which density parameter truly determines the approximation ratio?
Is it $\nabla_1$, or is it another parameter?

\paragraph{Why maximum average degree?}
Some familiar graph classes already give a glimpse of the answer.
The optimal ratio is $3$ for trees\footnote{The algorithm is folklore: just select every non-leaf if the graph is not an edge. This gives a $3$-approximation.} and $5$ for outerplanar graphs~\cite{BonamyEtAl2021Outerplanar}, while a ratio below $7$ is impossible on planar graphs~\cite{HilkeEtAl2014}.
Now compare this to the maximum average degree $\mad$ in these three classes. Very quickly, $\mad(G)=\max_{\varnothing\neq H\subseteq G}2|E(H)|/|V(H)|$ is the largest average degree of a nonempty subgraph of $G$, and it can also be defined on a class of graphs by taking the supremum.
For trees, outerplanar and planar graphs, we get a maximum average degree of $2$, $4$, and $6$, respectively.
These numbers suggest that the right approximation factor should be, very naturally, the maximum average degree of the class plus one.

The lower bound techniques currently known for approximation ratios also give the same intuition through a different reasoning.
Roughly, these techniques consist of building a graph where the view of almost every vertex is the same in any constant-round \LOCAL algorithm.
Therefore, almost every vertex makes the same choice.
They have two possible choices: not being included in the returned set, or being included.
Because we require the returned set to be dominating, this forces us to select almost all vertices.\footnote{Here, we leave out the construction of the identifier assignment. However, this can be made formal with the help of Ramsey's theorem.}
Now consider the power of a path $P_n^r$ example with $d=2r$. Almost all vertices have degree $d$.
Choose one vertex every $d+1$ vertices in our a dominating set (apart from constant-size intervals at both ends). 
This gives a dominating set of size at most $|V(G)|/(d+1)+O(1)$.
Conversely, every dominating set has size at least $n/(d+1)$, since each vertex can only dominate at most $d+1$ vertices.
Hence, we get a lower bound $d+1$~\cite{Picavet2026Thesis,AmiriSiebertz2026}. This is also the best that can be achieved with this technique, without developing significant new ideas.

\paragraph{The role of shallow-minor density.}
Thus, maximum average degree seems to be an important ingredient of the final answer.
However, having bounded maximum average degree alone is not enough to admit a constant-factor constant-round approximation.
Indeed, Dory, Ghaffari, and Ilchi~\cite{DoryEtAl2024} showed that even graphs of arboricity two (and thus maximum average degree less than $4$) do not admit a constant-factor approximation in a constant number of rounds.
Therefore, the question is now: what is the right additional property that is needed, on top of a bound on the maximum average degree?
Heydt et al.~\cite{HeydtEtAl2025} had already shown that bounding $\nabla_1$ is enough to obtain such an algorithm.
This also suggests that a bound on $\nabla_1$ is part of the answer. Even more interestingly, this can be in some way thought of as tight.
Indeed, relaxing this $\nabla_1$ to its most natural weakening, namely $\nabla_0$, which is half the maximum average degree, is not possible: the lower bound above rules it out.

\paragraph{Our contributions.}
What we show is that the bound on $\nabla_1$ is only needed to bound the running time, and our approximation ratio is almost what we conjectured before: $\mad+1+\eps$ can be achieved, for any $\eps>0$. The running time of course depends on this parameter $\eps$.
Moreover, our algorithm is almost tight for every even integer $d\geq2$. 
This argument is made possible by considering powers of paths, where the best ratio that can be obtained is $d+1$, even on bounded-expansion classes of maximum average degree at most $d$~\cite{Picavet2026Thesis,AmiriSiebertz2026}.
Thus, our algorithm is tight-up-to-$\eps$.

\begin{restatable}{corollary}{madcorollary}\label{cor:mad}
    For fixed $d,\beta\geq1$ and $\eps>0$, there is a deterministic constant-round \LOCAL algorithm that computes a $(d+1+\eps)$-approximation of \MinDS on graphs $G$ with $\mad(G)\leq d$ and $\nabla_1(G)\leq\beta$.
    In particular, this algorithm applies to every bounded expansion class of maximum average degree at most $d$.
    For $\eps>0$, the round complexity is $O_{d,\beta}(\eps^{-4}\log(1/\eps)^2)$, and the local computation is polynomial for fixed $d,\beta,\eps$.
\end{restatable}
This is the case $b=0$ of a slightly more general theorem, \cref{thm:main}.
For planar graphs, this gives a $7+\eps$ ratio, and for triangle-free planar graphs it gives $5+\eps$, matching the results of Amiri and Siebertz~\cite{AmiriSiebertz2026} with a faster run time.
Both approximation ratios are also generalized to the bounded genus setting.

Our algorithm performs local computation that is polynomial in the network size for fixed $d,\beta$, and $\eps$, but its exponent may depend on these parameters.
\Cref{fig:class-containments} shows how the graph classes covered by these results relate to one another.

This does not mean that $d+1$ is the best ratio that can be achieved for every fixed graph class.
For graphs of girth at least seven, Lenzen and Wenning~\cite{LenzenWenning2024} give a three-round algorithm with approximation ratio $O(\sqrt{\nabla_1})$.
On the high-girth members of a proper minor-closed class $\cC$, this is also sublinear in $\mad(\cC)$, since $\nabla_1(\cC)$ and $\nabla_0(\cC)$ are equal on that class.
Another example is the $50$-approximation of Bonamy et al.~\cite{BonamyEtAl2025ExcludedMinors} for $K_{2,t}$-minor-free graphs, whose ratio is independent of $t$, even though in this class, the maximum average degree is $t+1$ for $t\geq2$~\cite{ChudnovskyEtAl2011}.
Thus, additional structure can sometimes give better approximation ratios.
However, for bounded-expansion classes in general, our bound is tight up to $\eps$ for every even integer $d\geq2$.

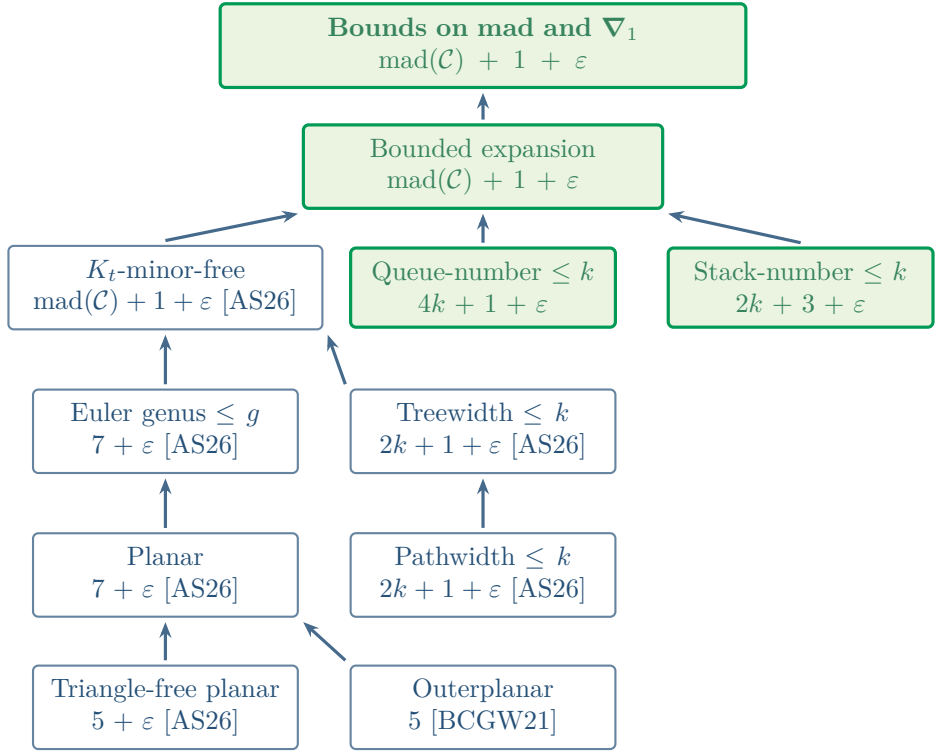
\begin{figure}[H]\providecommand{\rewritewording}[1]{#1}\providecommand{\unchangedwording}[1]{#1}
    \centering
    \begin{tikzpicture}[
        classbox/.style={rounded corners=2pt,draw=primary!70,line width=.8pt,text=primary,align=center,font=\small,text width=3.15cm,minimum height=1.03cm,inner sep=5pt},
        newclass/.style={classbox,draw=mydarkgreen,fill=mygreen!18,text=mydarkgreen!70!black,line width=1.2pt},
        every node/.style={outer sep=2pt}
    ]
        \node[newclass,text width=6.6cm] (main) at (0,4.15)
            {\textbf{Bounds on }$\bm\mad$\textbf{ and }$\bm\nabla_1$\\
             $\mad(\cC)+1+\eps$};
        \node[newclass,text width=4.4cm] (expansion) at (0,2.55)
            {Bounded expansion\\$\mad(\cC)+1+\eps$};

        \node[classbox,text width=3.8cm] (minor) at (-4.2,.95)
            {$K_t$-minor-free\\$\mad(\cC)+1+\eps$~\cite{AmiriSiebertz2026}};
        \node[newclass] (queue) at (0,.95)
            {Queue-number $\leq k$\\$4k+1+\eps$};
        \node[newclass] (stack) at (4.2,.95)
            {Stack-number $\leq k$\\$2k+3+\eps$};

        \node[classbox] (genus) at (-4.2,-.95)
            {Euler genus $\leq g$\\$7+\eps$~\cite{AmiriSiebertz2026}};
        \node[classbox] (treewidth) at (0,-.95)
            {Treewidth $\leq k$\\$2k+1+\eps$~\cite{AmiriSiebertz2026}};
        \node[classbox] (planar) at (-4.2,-2.85)
            {Planar\\$7+\eps$~\cite{AmiriSiebertz2026}};
        \node[classbox] (pathwidth) at (0,-2.85)
            {Pathwidth $\leq k$\\$2k+1+\eps$~\cite{AmiriSiebertz2026}};
        \node[classbox] (trianglefree) at (-4.2,-4.6)
            {Triangle-free planar\\$5+\eps$~\cite{AmiriSiebertz2026}};

        \node[classbox] (outerplanar) at (0,-4.6)
            {Outerplanar\\$5$~\cite{BonamyEtAl2021Outerplanar}};

        \draw[edge,-{Stealth[length=2.1mm,width=1.6mm]},draw=primary!85] (expansion.north) -- (main.south);
        \draw[edge,-{Stealth[length=2.1mm,width=1.6mm]},draw=primary!85] (minor.north) -- (expansion.south west);
        \draw[edge,-{Stealth[length=2.1mm,width=1.6mm]},draw=primary!85] (queue.north) -- (expansion.south);
        \draw[edge,-{Stealth[length=2.1mm,width=1.6mm]},draw=primary!85] (stack.north) -- (expansion.south east);
        \draw[edge,-{Stealth[length=2.1mm,width=1.6mm]},draw=primary!85] (genus.north) -- (minor.south);
        \draw[edge,-{Stealth[length=2.1mm,width=1.6mm]},draw=primary!85] (treewidth.north west) -- (minor.south east);
        \draw[edge,-{Stealth[length=2.1mm,width=1.6mm]},draw=primary!85] (planar.north) -- (genus.south);
        \draw[edge,-{Stealth[length=2.1mm,width=1.6mm]},draw=primary!85] (pathwidth.north) -- (treewidth.south);
        \draw[edge,-{Stealth[length=2.1mm,width=1.6mm]},draw=primary!85] (trianglefree.north) -- (planar.south);
        \draw[edge,-{Stealth[length=2.1mm,width=1.6mm]},draw=primary!85] (outerplanar.north west) -- (planar.south east);
    \end{tikzpicture}
    \caption{
        Some graph classes covered by our results and their respective approximation ratios.
        An arrow from class $\cC$ to class $\cD$ means $\cD$ contains $\cC$.
        Our new results are shown in green.
        Note that for $3\leq t\leq9$, we can get a ratio of $2t-3+\eps$ for $K_t$-minor-free graphs.
    }
    \label{fig:class-containments}
\end{figure}

\subsection{Related works}

Constant-round approximation algorithms of \MinDS in the distributed setting have been widely studied.
In this section, we give an overview of a selection of results in this field of study.
We focus on constant-round algorithms in the \LOCAL model whose approximation ratio is also independent of the maximum degree.

On general graphs, Kuhn, Moscibroda, and Wattenhofer~\cite{KuhnEtAl2016} proved that a constant-factor approximation for \MinDS requires $\Omega\bigl(\min\{\log\Delta(G)/\log\log\Delta(G),\allowbreak\sqrt{\log n/\log\log n}\}\bigr)$ \LOCAL rounds.
Thus, we need to restrict the input graph class to get deterministic constant-round algorithms with a constant approximation factor.
This is the main setting considered for the rest of this subsection.

\paragraph{The history of approximations for planar graphs.}
Our story begins with planar graphs, for which Lenzen, Pignolet, and Wattenhofer~\cite{LenzenEtAl2013} gave a $130$-approximation algorithm.
This ratio was later improved to $52$ by Wawrzyniak~\cite{Wawrzyniak2014} by giving a stronger analysis of the same algorithm.
Then, by designing a different algorithm, Heydt, Siebertz, and Vigny~\cite{HeydtEtAl2022,HeydtEtAl2025} were able to achieve $11+\eps$ for positive $\eps$.
Very recently, Wawrzyniak~\cite{Wawrzyniak2026} obtained $8+\eps$, which was improved shortly afterwards by the $7+\eps$ bound of Amiri and Siebertz~\cite{AmiriSiebertz2026}.

\paragraph{Subclasses of planar graphs.}
Subclasses of planar graphs have also been studied. For instance, Wawrzyniak~\cite{Wawrzyniak2025} gave a constant-round $6$-approximation for triangle-free planar graphs. Very recently, Amiri and Siebertz~\cite{AmiriSiebertz2026} achieved the tight-up-to-$\eps$ factor $5+\eps$ as predicted by Bonamy et al.~\cite[Section~4]{BonamyEtAl2021Outerplanar}.
For outerplanar graphs, Bonamy et al.~\cite{BonamyEtAl2021Outerplanar} obtained a tight factor $5$.
In terms of lower bounds, Hilke, Lenzen, and Suomela~\cite{HilkeEtAl2014} proved that for every $\eps>0$, no deterministic constant-round $(7-\eps)$-approximations for planar graphs exist in a constant number of rounds.
Recently, Amiri and Siebertz~\cite[Theorem~8.7]{AmiriSiebertz2026} subsequently proved that no constant-round algorithm can beat the ratio $2r+1$ on the $r$-th powers of paths.
The same $2r+1$ lower bound had already been proved before in the author's thesis~\cite[Lemma~V.2.1]{Picavet2026Thesis} for $r$-th powers of paths with additional leaves attached to one endpoint. 

\paragraph{Sparse graph classes.}
Various generalisations to sparse graph classes have also been studied.
Czygrinow et al.~\cite{CzygrinowEtAl2018} gave a constant-round constant-approximation algorithm for classes excluding a topological minor.
Kublenz, Siebertz, and Vigny~\cite{KublenzEtAl2021} then extended this to bounded-expansion classes.
This result was subsequently improved by Heydt et al.~\cite{HeydtEtAl2025} to the larger class of graphs of bounded $\nabla_1(G)$.
We also work in the same bounded $\nabla_1(G)$ setting. \Cref{cor:mad} gives an approximation ratio depending only on the maximum average degree of the input class.
Note that bounded arboricity (i.e. bounded $\nabla_0(G)$) alone is not sufficient for a constant-factor approximation in constant rounds.
Dory, Ghaffari, and Ilchi~\cite{DoryEtAl2024} proved that the lower bound of Kuhn, Moscibroda, and Wattenhofer also holds on graphs of arboricity two.
Thus, constant-factor approximations require $\Omega\bigl(\min\{\log\Delta(G)/\log\log\Delta(G),\allowbreak\sqrt{\log n/\log\log n}\}\bigr)$ deterministic rounds even on these graphs.
For non-constant rounds, $O(\log^* n)$ rounds are sufficient for $(1+\eps)$-approximations on planar graphs~\cite{CzygrinowEtAl2008} and, more generally, for graphs excluding a fixed minor~\cite{CzygrinowEtAl2018}.

\paragraph{Truly constant approximations.}
Recently, some work has been done on algorithms whose approximation ratio does not depend on the parameters of the class.
Bonamy, Gavoille, Picavet, and Wesolek~\cite{BonamyEtAl2025ExcludedMinors} proved a $50$-approximation on $K_{2,t}$-minor-free graphs running in $O_t(1)$ rounds, breaking the linear-in-$t$ barrier for the approximation ratio.
Bonamy, Gavoille, Das, Suomela, Picavet, and Wesolek~\cite{BonamyEtAl2026} then obtained a $(34+\eps)$-approximation on graphs of Euler genus at most $g$ in $O_{g,\eps}(1)$ rounds (without the need to be given the embedding).
Here too, the approximation ratio is independent of the parameters of the class, while the round complexity depends on them.
The authors proved these results by developing a series of meta-theorems for a particular class of minimization problems on special graph classes.
For \MinDS, these automatically turn any constant-round planar $\alpha$-approximation into a $(3\alpha+1)$-approximation on graphs of bounded Euler genus.
The round complexity depends only on the genus and the round complexity of the planar algorithm.
Their meta-theorems also apply to other problems, such as \textsc{Minimum $k$-tuple Dominating Set} or \textsc{Distance-$k$ Minimum Dominating Set}.

More recently, Wawrzyniak~\cite{Wawrzyniak2026Transfer} refined these meta-theorems, achieving arbitrarily small loss $\alpha+\eps$, at the cost of asking for assumptions at a larger scale that depends on $\eps$.
For \MinDS, this new meta-theorem takes a constant-round planar $\alpha$-approximation as input and turns it into an $(\alpha+\eps)$-approximation on graphs of bounded Euler genus.
Again, the round complexity only depends on the genus, $\eps$, and the round complexity of the planar algorithm~\cite[Sections~7.1--7.3]{Wawrzyniak2026Transfer}.
With the help of this meta-theorem and the planar and triangle-free planar algorithms of Amiri and Siebertz, one can get $(7+\eps)$- and $(5+\eps)$-approximations on bounded-genus graphs and triangle-free bounded-genus graphs, respectively.

\paragraph{Other models.}
In anonymous planar networks with port numbering, where nodes do not have identifiers, Wawrzyniak~\cite{Wawrzyniak2015} gave a deterministic constant-factor approximation in constant rounds, that uses short messages.
Some other algorithms were also studied in the context of the \CONGEST model, which is very similar to \LOCAL, only that each edge can only transfer $O(\log n)$ bits per round.
Czygrinow et al.~\cite{CzygrinowEtAl2019} designed a constant-factor constant-round approximation algorithm for graphs of bounded orientable genus in \CONGEST\textsubscript{\textsf{BC}}, a broadcast variant of \CONGEST where each vertex has to send the same message to all its neighbors.
Later, Dory, Ghaffari, and Ilchi~\cite{DoryEtAl2024} studied graphs of arboricity at most $\alpha$, and gave a deterministic $(2\alpha+1)(1+\eps)$-approximation for weighted \MinDS in $O(\eps^{-1}\log\Delta(G))$ \CONGEST rounds.
More recently, Alipour and Farokhnejad~\cite{AlipourFarokhnejad2026} obtained a deterministic $O(\alpha\log\Delta(G)/\log\log\Delta(G))$-approximation in $O(\log\Delta(G)/\log\log\Delta(G))$ rounds. The algorithm needs to send only one bit per edge per round, but requires knowledge of $\Delta(G)$.

\subsection{Comparison with recent algorithms}\label{sec:comparison}

In this section, we compare the very recent constant-round approximation algorithms for \MinDS in the \LOCAL model with our algorithm.
Our work was initiated and results were obtained independently before the release of the papers by Amiri and Siebertz~\cite{AmiriSiebertz2026} and Wawrzyniak~\cite{Wawrzyniak2026}.
However, our algorithm draws lots of its inspiration from the $(11+\eps)$-approximation algorithm of Heydt et al.~\cite{HeydtEtAl2025} and preprocesses the graph before solving an LP and then computing the final dominating set via rounding.
The key difference in our work is the idea of introducing hypergraph constraints in the LP to represent choices for ``hard to dominate'' regions, and then of bounding this hypergraph's degree using the sunflower lemma.
We also have to argue that solving this new LP instead of the standard domination LP does not increase the solution size too much.
Our strategy differs from that of Amiri and Siebertz, who use a clustering and averaging approach, and from that of Wawrzyniak, who improves the analysis and slightly changes the algorithm of Heydt et al.
It also applies to more general graph classes: \cref{cor:mad} only requires bounded maximum average degree and depth-one shallow-minor density.
Thus, our algorithm works on proper minor-closed classes, but also on some classes not contained in any proper minor-closed class, as illustrated in \cref{cor:other-classes}.

In the rest of the subsection, we sometimes refer to $\MDS(G)$, the minimum cardinality of a dominating set, as the integral optimum.
In contrast, the fractional optimum $\MDS_f(G)$ is the minimum of $\sum_{v\in V(G)}x_v$ over nonnegative vertex weights $x_v$ satisfying $\sum_{v\in N_G[w]}x_v\geq1$ for every vertex $w$.
We say that an algorithm has a guarantee against the integral optimum if the returned dominating set has size at most $\alpha\MDS(G)$ for some $\alpha$, and we say it has a guarantee against the fractional optimum if the returned dominating set has size at most $\alpha\MDS_f(G)$ for some $\alpha$.\footnote{Note that, in both cases, the algorithm outputs an actual dominating set.}
In the latter, we compare to the optimal size of a solution to the fractional LP. Thus $\MDS_f(G)\leq\MDS(G)$, and the guarantee against the fractional optimum is a stronger conclusion than the guarantee against the integral optimum, if the ratio $\alpha$ is the same in both cases.
\Cref{tab:planar-comparison} summarizes the graph classes covered by the algorithms and their guarantees for $\eps>0$.

\begin{table}[htbp]
    \centering
    \small
    \renewcommand{\arraystretch}{1.15}
    \setlength{\tabcolsep}{4pt}
    \begin{tabular}{@{}l|l|c|c@{}}
    \hline
    Result & Graph class $\cC$ & MDS of size at most & Number of rounds \\
    \hline
    Wawrzyniak~\cite{Wawrzyniak2026} & Planar & $(8+\eps)\MDS(G)$ & $O_\eps(1)$ \\
    & Euler genus $\leq g$ & $(25+\eps)\MDS(G)$ & $O_{g,\eps}(1)$ \\
    Amiri--Siebertz~\cite{AmiriSiebertz2026} & Excluding a fixed minor & $(2a(\cC)+1+\eps)\MDS_f(G)$ & $O_{\cC}(\eps^{-\kappa_a})$ \\
    \Cref{cor:mad} & Bounded $\mad$ and $\nabla_1$ & $(\mad(\cC)+1+\eps)\MDS(G)$ & $O_{\cC}(\eps^{-4}\log(1/\eps)^2)$ \\
    \hline
    \end{tabular}
    \caption{Bounds on the size of the returned dominating set for different graph classes and different algorithms. Here $\mad(\cC)=\sup_{G\in\cC}\mad(G)$ and $a(\cC)$ is the supremum of the edge-to-vertex ratio over all graphs of the class. Note that for classes closed under taking subgraphs, $\mad(\cC)=2a(\cC)$. The constant $\kappa_a$ is equal to $\ln 3/\ln(8a/(8a-1))$. For planar graphs, we have $a=3$, so $\kappa_3<25.9$.}
    \label{tab:planar-comparison}
\end{table}

\paragraph{Amiri and Siebertz.}
For a fixed proper minor-closed class $\cC$, let
\[
    a(\cC)=\sup\left\{\frac{|E(H)|}{|V(H)|}~\middle|~ H\in\cC \text{ non-empty}\right\}.
\]
Amiri and Siebertz~\cite[Theorem~1.1]{AmiriSiebertz2026} give a $(2a(\cC)+1+\eps)$-approximation against the fractional optimum, in constant rounds.
With this algorithm, they are able to derive a $(5+\eps)$-approximation for triangle-free planar graphs against the fractional optimum, and a $(7+\eps)$-approximation for bounded Euler-genus graphs against the integral optimum~\cite[Corollaries~1.4 and~8.2]{AmiriSiebertz2026}.
Both of those ratios are tight up to $\eps$.
Our coefficient agrees with theirs when $d=2a(\cC)$. However, we instead compare to the size of the integral solution and not the fractional solution. Note that for any minor-closed class $\cC$ (and more generally, any monotone class), it is true that we can choose $d=2a(\cC)$.
Thus, while we have a weaker guarantee, our algorithms cover graph classes not covered by the algorithm of Amiri and Siebertz.
Indeed, our algorithms apply in classes not included in any proper minor-closed class, such as bounded queue-number graphs or bounded stack-number graphs. See \cref{cor:other-classes} for the full results.

Moreover, our algorithms have a better running time dependency in $\eps$, and the local computation of our algorithm is polynomial in the size of the view of every vertex, for fixed parameters.
This addresses the efficiency questions in~\cite[Section~9]{AmiriSiebertz2026}, in the integral-optimum case.

\paragraph{Wawrzyniak.}
Wawrzyniak~\cite{Wawrzyniak2026} obtains $8+\eps$ by strengthening the structural analysis of Heydt et al. The author also derives a $(25+\eps)$-approximation for bounded Euler genus~\cite[Corollary~7.1]{Wawrzyniak2026}. We improve this ratio to $7+\eps$ in \cref{cor:genus}, matching the integral bounds of Amiri and Siebertz.

\paragraph{Heydt et al.}
As said before, our algorithm preprocesses the graph before solving and rounding an LP, like in the algorithm of Heydt et al.~\cite{HeydtEtAl2025}.
We adapt some of their results, such as their reduction to a smaller graph, their selection of high-degree vertices, and their LP rounding, to our case.
The difference is that we encode costly domination choices as additional constraints in the LP, and then use sunflowers to make sure this LP can be solved quickly.
Both our algorithms and theirs apply to the same graph classes, i.e. when $\nabla_1$ is bounded. 
However, we improve the approximation ratio so that ours depends only on the maximum average degree of the input class.

\Needspace{10\baselineskip}
\subsection{High-level overview}\label{sec:overview}

Very roughly, our algorithm is divided into four stages: first, we build a hypergraph, whose edges have a common neighborhood ``hard to dominate'' without taking a vertex of the hyperedge itself.
Second, we take in our solution a very small set of vertices that dominates many vertices, and forget about already-dominated vertices.
Third, we compute a fractional solution to the slightly stronger problem of finding a fractional hitting set of this hypergraph, while also fractionally dominating the other vertices not fractionally dominated directly by the fractional hitting set.
Then, we finally round this fractional solution to get an actual dominating set of the graph.
We give more detail in the rest of this subsection.

\paragraph{Rounding.}
Suppose we are in step four above: we have a fractional solution to an LP that dominates the vertices we still need to dominate and hits every hyperedge we still need to hit. First suppose that $\mad(G)\leq d$.
In \cref{sec:rounding}, we prove that we can round this to an integral solution and get a dominating set of size at most $d+1$ times the original weight of the fractional solution.\footnote{Note that the factor $d+1$ appears only because of the rounding procedure. All other algorithms have arbitrarily small loss.}
To do this, we first select all vertices that have weight at least $1/(d+1)$. This forms some set of vertices. We also add to the returned set all vertices undominated by this former set.
This rounding procedure still holds under the weaker assumption that $|E(H)|\leq(d/2)|V(H)|+b$ for every subgraph $H$, at the cost of an additional $4b/d$ in the size of the returned set (see \cref{lem:rounding}).
This is the rounding procedure used in Heydt et al.~\cite[Section~6.1]{HeydtEtAl2025}, adapted here to additionally satisfy the hypergraph hitting-set constraints and to work under the weaker assumption described above.

\paragraph{Preparing a solution that can be rounded.}
We give a short summary of the algorithm of \cref{thm:main} and its proof. See \cref{alg:main} for the pseudocode.
Fix an accuracy parameter $0<\rho\leq1$, that will be chosen sufficiently small with respect to $d$ and $\eps$.
All constants we give below may depend on $d$, $\beta$, $b$, and $\rho$. They do not depend on the size of the graph.

\begin{enumerate}
    \item For any set $A$, define $X(A)=\bigcap_{v\in A}N_G(v)$.
    Construct a hypergraph $\cH_0$ whose hyperedges are all small enough, nonempty sets $A$ for which $X(A)$ cannot be dominated with a small number of vertices outside of $A$. This gives a hypergraph of bounded rank.
    We sparsify it using the sunflower lemma to get a hypergraph $\cH$ of bounded maximum degree.
    Throughout this sparsification, we also keep a minimum dominating set that can be enlarged to a set $C$ that is at most a $\rho$-fraction bigger and hits every hyperedge of $\cH$.
    Note that $C$ is not computed in the algorithm. It is only used in the analysis to prove upper bounds on the size of the returned set (see \cref{sec:auxiliary,sec:hypergraph}).

    \item A hitting set of $\cH$ already dominates the vertices in $X=\bigcup_{e\in E(\cH)}X(e)$.
    The remaining set $R_0=V(G)\setminus X$ has size $O(\MDS(G))$, but is still too large for us to simply include all of it into our output.
    However, we prove a very important property: every vertex has only a bounded number of neighbors in $R_0$ (see \cref{sec:residual}).

    \item Thus, we should only select \textit{some} vertices of $R_0$. 
    We locally compute the subgraph $G'$ induced by the vertices of $V(\cH)\cup R_0$ and, among the other vertices, one representative $v$ for each different nonempty $N_G(v)\cap R_0$. This subgraph has size $O(\MDS(G))$.
    We add to our output the set $P$ of vertices of sufficiently high degree in $G'$.
    This allows us to bound the maximum degree of $F=G'-P$. Moreover, by construction of $G'$, taking this set $P$ only costs $\rho\MDS(G)$.
    Moreover, we can construct a set $C'\subseteq V(G')$ from $C$ that still dominates $R_0$ and is a hitting set of $\cH$ in $G'$ (see \cref{sec:lp-preparation}).

    \item Solve an LP on $F$ that fractionally hits hyperedges not already hit by $P$ and fractionally dominates vertices of $R_0$ not already dominated by $P$ (see \cref{sec:locality}).
    Because of the bounded rank and degree of the hypergraph, and the bounded degree of $F$, we can prove that the row and column sums of the LP are bounded.
    So, there exists a constant round \LOCAL algorithm that gives a fractional solution a $\rho$-fraction bigger than its optimum.
    Since $C'\setminus P$ is a solution to the LP, its fractional solution's weight can be bounded by $(1+\rho)^2\MDS(G)$.
    Round this fractional solution and add $P$ to get our returned dominating set.

    \item This already gives a $(\mad(\cC)+1+\eps)$-approximation on the class $\cC$. However, if the graph does not satisfy $\mad(G)\leq d$ but instead we have $|E(H)|\leq\frac d2|V(H)|+b$ for every subgraph $H$ of $G$, then remove the additive error term $b$ by solving graphs with small domination number exactly. 
    This is the reduction idea of Bonamy et al.~\cite[Proposition~D.2]{BonamyEtAl2026}.
    Such graphs have bounded diameter and can be solved in constant rounds.
    On all other graphs, the additive $b$ term is very small compared to $\MDS(G)$ (see \cref{sec:additive,sec:assembly}).
\end{enumerate}

The third step above, where we work in $G'$ and select its high-degree vertices, is adapted from Heydt et al.~\cite[Sections~6.3--6.4]{HeydtEtAl2025}. Here, we additionally preserve the hyperedge constraints.
The new idea is to combine all the hitting set constraints, sparsified in the first step, with the domination constraints into a single LP. 
The sparsification step allows us to bound the row/column sums of the LP. This guarantee is needed to apply the LP algorithm of Kuhn, Moscibroda, and Wattenhofer~{\cite[Corollary~4.1]{KuhnEtAl2006}}.

\begin{algorithm}[H]
    \caption{Dominating set approximation}\label{alg:main}
    \KwIn{A graph $G$ satisfying the density assumptions of \cref{thm:main}, and $\eps>0$.}
    \KwOut{A dominating set of size at most $(d+1+\eps)\MDS(G)$.}
    Set $\rho\gets\min\{1,\eps/(3(d+1)+2)\}$ and use $r,q,t,\Gamma$ as in \cref{eq:structural-parameters,eq:cutoff}\;
    Set $h\gets\max\{1,\lceil4b/(d\rho)\rceil\}$\;
    \If{the test of \cref{lem:additive-removal} finds $\MDS(G)\leq h$}{
        \Return{a canonical minimum dominating set of the graph}\;
    }
    Build $\cH_0$ from the ``hard to dominate'' common neighborhoods and sparsify it to $\cH$ as in \cref{claim:sunflower}\;
    Set $X\gets\bigcup_{e\in E(\cH)}X(e)$ and $R_0\gets V(G)\setminus X$\;
    Compute $G'$ by keeping $V(\cH)\cup R_0$ and representatives of the other nonempty $R_0$-neighborhoods\;
    Set $P\gets\{v\in V(G')\mid d_{G'}(v)>\Gamma\}$ and $F\gets G'-P$\;
    Set $R_1\gets R_0\setminus N_G[P]$ and $\cH'\gets\{e\in E(\cH)\mid e\cap P=\varnothing\}$\;
    Compute the fractional solution $x$ of \cref{lem:local-lp}\;
    Run \cref{alg:rounding} on $(F,R_1,x,d)$, obtaining $A\cup U$\;
    \Return{$P\cup A\cup U$}\;
\end{algorithm}

\section{Preliminaries}\label{sec:preliminaries}

\paragraph{Graphs.}
All graphs are finite, simple, and undirected, and the algorithms are assumed to run on connected graphs.
We write $V(G)$ and $E(G)$ for the vertex set and edge set of $G$, respectively, and use $n=|V(G)|$ when no confusion is possible.
The \emph{open neighborhood} $N_G(v)$ of a vertex $v$ is its set of neighbors, and its \emph{closed neighborhood} is $N_G[v]=N_G(v)\cup\{v\}$.
The open neighborhood of a set of vertices $U$ is $N_G(U)=\left(\bigcup_{u\in U}N_G[u]\right)\setminus U$, and its closed neighborhood is $N_G[U]=U\cup N_G(U)$.
A \emph{hypergraph} $\cH$ consists of a vertex set $V(\cH)$ and a family $E(\cH)$ of subsets of $V(\cH)$, called \emph{hyperedges}. Its \emph{rank} is the maximum size of a hyperedge, or zero if there are no hyperedges, and the \emph{degree} of a vertex is the number of hyperedges containing it.
A \emph{hitting set} of $\cH$ is a set of vertices intersecting every hyperedge. 
We say that a subset of vertices \emph{hits} $\cH$ if it is a hitting set of $\cH$.

\paragraph{Graph classes and parameters.}

A graph is \emph{planar} if it can be drawn in the plane without crossing its edges.
A surface has Euler genus $2g$ if it is orientable of genus $g$, and $g$ if it is nonorientable of genus $g$. 
The \emph{Euler genus} of a graph is the minimum $g$ such that it can be embedded into a closed connected surface of Euler genus $g$ without crossing its edges.
In particular, planar graphs have Euler genus zero.
A \emph{minor} of $G$ is a simple graph obtained from $G$ by deleting some of its vertices or edges, and by contracting some edges.
A graph is said to be \emph{$H$-minor-free} if it does not contain $H$ as a minor, and a class \emph{excludes a fixed minor} if there is a graph $H$ such that every graph of the class is $H$-minor-free.

\paragraph{Sparse graphs.}

We give a few different parameters measuring the sparsity of a graph class.
The \emph{density} of a nonempty graph $G$ is defined as $|E(G)|/|V(G)|$, i.e. its edge-to-vertex ratio.
Its \emph{maximum average degree} $\mad(G)$ is the maximum of $2|E(H)|/|V(H)|$ over all nonempty subgraphs $H\subseteq G$, i.e. twice the density of its densest subgraph.
The \emph{arboricity} $\arb(G)$ is the minimum number of forests whose edge sets partition $E(G)$, and the \emph{degeneracy} $d(G)$ of $G$ is the largest minimum degree of a nonempty subgraph of $G$.
$\mad, \arb$ and $d$ are all linearly equivalent.
Indeed, we have the inequality $\mad(G)/2\leq \arb(G)\leq d(G)\leq\mad(G)$. Thus, bounding one of these parameters bounds all the others.
For an integer $r$, a \emph{depth-$r$ shallow minor} of $G$ is a simple graph obtained from a subgraph of $G$ by contracting pairwise vertex-disjoint connected subgraphs of radius at most $r$, with radius defined inside each contracted subgraph. For $r=1$, each contracted subgraph contains a spanning star.
$\nabla_r(G)$ is defined as the maximum density $|E(H)|/|V(H)|$ of a nonempty depth-$r$ shallow minor $H$ of $G$.
Thus, we have $\nabla_0(G)=\mad(G)/2$. $\nabla_1(G)$ is the \emph{depth-one shallow-minor density} used in our theorem.
A graph class $\cC$ has \emph{bounded expansion} if there is a function $f$ such that $\nabla_r(G)\leq f(r)$ for every nonnegative integer $r$ and graph $G\in\cC$~\cite{NesetrilEtAl2012}.
In particular, every bounded-expansion class has bounded $\nabla_1$. Therefore, our theorem also applies to every bounded-expansion class.

\paragraph{The \LOCAL model.}
In the deterministic \emph{\LOCAL model}, each vertex is a processor with a unique identifier taken from the set $[1,n^c]$ where $c$ is some constant.
Communication works in synchronous rounds. During each round, vertices can perform arbitrary local computation and then send messages of arbitrary size to their neighbors.
Therefore, for an algorithm running in $T$ rounds, the output of each vertex can only depend on the graph induced by its $T$-radius neighborhood and the associated identifiers.
The \emph{round complexity} is the number of communication rounds needed for every vertex to output.

\section{The algorithm and its analysis}\label{sec:algorithm}

The following theorem allows an additional additive bound in the edge-density. \Cref{cor:mad} then follows by setting $b=0$.
The parameter $d$ determines the approximation ratio, while $\beta$ and $b$ are only here to bound the running time.

\begin{theorem}\label{thm:main}
    For fixed $d,\beta\geq1$, $b\geq0$, and $\eps>0$, there is a deterministic constant-round \LOCAL algorithm that computes a $(d+1+\eps)$-approximation of \MinDS on graphs $G$ with $\nabla_1(G)\leq\beta$ and
    \begin{align*}
        |E(H)|\leq\frac d2|V(H)|+b\quad\forall H\subseteq G.
    \end{align*}
    For $\eps>0$, the round complexity is $O_{d,\beta,b}(\eps^{-4}\log(1/\eps)^2)$, and the local computation is polynomial for fixed $d,\beta,b,\eps$.
\end{theorem}

\subsection{The rounding procedure}\label{sec:rounding}

In this subsection, we adapt the rounding procedure used by Heydt et al.~\cite[Section~6.1]{HeydtEtAl2025} to our setting.
Given a subset of vertices to dominate $R$, \cref{alg:rounding} selects vertices of sufficiently large weight, then adds any vertices of $R$ that are still undominated.

\bigskip

\begin{algorithm}[H]
    \caption{Rounding procedure}\label{alg:rounding}
    \KwIn{A graph $G$, a set $R\subseteq V(G)$, non-negative vertex weights $x$, and $d>0$.}
    \KwOut{A set $A\cup U$ of selected vertices.}
    $A\gets\{v\in V(G)\mid x_v\geq1/(d+1)\}$\;
    Each vertex informs all its neighbors whether it belongs to $A$\;
    $U\gets R\setminus N_G[A]$\;
    \Return{$A\cup U$}\;
\end{algorithm}

\bigskip

The following lemma bounds the size of the returned set using the graph's density, and additionally shows that we can satisfy hyperedge constraints at the same time.

\begin{lemma}\label{lem:rounding}
    Let $d>0$, $b\geq0$, and $G$ be a graph such that $|E(G[X])|\leq(d/2)|X|+b$ for every $X\subseteq V(G)$.
    Let $R\subseteq V(G)$, and let $\cH$ be a hypergraph on $V(G)$ of rank at most $d+1$.
    Suppose some $x\in \R_{\geq 0}^{V(G)}$ satisfies
    \[
        \sum_{v\in N_G[u]}x_v\geq1\quad\forall u\in R
    \]
    and
    \[
        \sum_{v\in e}x_v\geq1\quad\forall e\in E(\cH).
    \]
    Then \cref{alg:rounding} returns a set $A\cup U$ that:
    \begin{enumerate}
        \item\label{item:rounding-domination} dominates $R$,
        \item\label{item:rounding-hitting} intersects every hyperedge of $\cH$, and
        \item\label{item:rounding-size} has size at most $(d+1)\sum_vx_v+4b/d$.
    \end{enumerate}
    This takes one deterministic \LOCAL round.
\end{lemma}
Notice that in the above lemma, one can take $b=0$ when $\mad(G)\leq d$.
Also notice that this is the main bottleneck where the $\mad(G)+1$ comes into play. The losses in the other steps can be made arbitrarily small.

\begin{proof}
The set $A\cup U$ dominates $R$ by definition of the set $U$.
This proves~\cref{item:rounding-domination}.
Every hyperedge $e\in E(\cH)$ contains a vertex of $A$, otherwise we would get a contradiction: $\sum_{v\in e}x_v<|e|/(d+1)\leq1$. 
This proves~\cref{item:rounding-hitting}.
It now remains to bound $|A\cup U|$.

For every vertex $v$ of $N_G[U]$, we have $x_v<1/(d+1)$ as $v\notin A$.
For every $T\subseteq V(G)$, write $x_T=\sum_{v\in T}x_v$.
Let us sum the constraints over $U\subseteq R$:
\[
    \sum_{u\in U}\sum_{v\in N_G[u]}x_v \geq |U|.
\]
Rewriting $\sum_{v\in N_G[u]}x_v$ as $\sum_{v\in N_G(u)}x_v+x_u$ gives
\begin{equation}\label{eq:rounding-demand}
    \sum_{u\in U}\sum_{v\in N_G(u)}x_v \geq |U| - x_U \geq \frac{d}{d+1}|U|
\end{equation}
as $(d+1)x_U\leq|U|$.
For $0\leq t\leq1/(d+1)$, define the sets $X_t=U\cup\{v\in N_G(U)\mid x_v\geq t\}$ and $Y_t=\{u\in U\mid x_u\geq t\}$.
In the following, we compute $\int_0^{1/(d+1)}|E(G[X_t])|+|E(G[Y_t])|\,dt$.
\Cref{fig:rounding-contributions} illustrates the two types of edge contributions that appear below.
An edge $uv$ of $G[U]$ contributes $1/(d+1)+\min\{x_u,x_v\}\geq x_u+x_v$ to the integral.
Indeed, $uv$ is in $G[X_t]$ for all $t\in[0,1/(d+1)]$, and is in $G[Y_t]$ if and only if $t\leq \min\{x_u,x_v\}$.
An edge $uv$ with $u\in U$ and $v\in N_G(U)$ contributes $x_v$: it is in $G[X_t]$ if and only if $t\leq x_v$ and is never in $G[Y_t]$.
All other contributions are non-negative: if both endpoints are in $N_G(U)$, the contribution is $\min\{x_u,x_v\}$. Otherwise, it is zero. Thus, we get:
\begin{align*}
    \sum_{u\in U}\sum_{v\in N_G(u)}x_v
    &\leq\int_0^{1/(d+1)}|E(G[X_t])|+|E(G[Y_t])|\,dt\\
    &\leq\int_0^{1/(d+1)}\frac d2|X_t|+\frac d2|Y_t|+2b\,dt\\
    &=\frac{d|U|}{2(d+1)}+\frac d2\sum_{v\in N_G[U]}\int_0^{1/(d+1)}\indicator_{\{t\leq x_v\}}\,dt+\frac{2b}{d+1}\\
    &=\frac{d|U|}{2(d+1)}+\frac d2\sum_{v\in N_G[U]}x_v+\frac{2b}{d+1}.
\end{align*}
Together with~\cref{eq:rounding-demand}, this implies $|U|\leq(d+1)x_{N_G[U]}+4b/d$ after rearranging.
Finally, $|A|\leq(d+1)x_A$ and $A\cap N_G[U]=\varnothing$, so
\[
    |A\cup U|\leq (d+1)x_A+(d+1)x_{N_G[U]}+4b/d
    \leq(d+1)x_{V(G)}+4b/d.
\]
This proves~\cref{item:rounding-size}.
\end{proof}

\begin{figure}[H]
    \centering
    \begin{tikzpicture}[x=1cm,y=1cm,font=\small,
        rd guide/.style={draw=primary!25,line width=.7pt},
        rd text/.style={text=primary,align=center}]
        \path[use as bounding box] (0,-1.03) rectangle (15.2,5.15);
        \foreach \shift in {0,7.8} {
            \begin{scope}[xshift=\shift cm]
                \draw[rd guide,-{Stealth[length=1.8mm]}] (2.0,.40) -- (6.65,.40);
                \node[rd text,anchor=north] at (2.0,.30) {$0$};
                \node[rd text,anchor=north] at (6.1,.30) {$1/(d+1)$};
                \node[rd text,anchor=west] at (6.65,.40) {$t$};
                \draw[rd guide,densely dotted] (2.0,.40) -- (2.0,2.0);
                \draw[rd guide,densely dotted] (6.1,.40) -- (6.1,2.0);
                \node[rd text,anchor=east,text=myblue] at (1.65,1.80) {$G[X_t]$};
                \node[rd text,anchor=east,text=myred] at (1.65,.95) {$G[Y_t]$};
                \draw[rd guide] (2.0,1.80) -- (6.1,1.80);
                \draw[rd guide] (2.0,.95) -- (6.1,.95);
            \end{scope}
        }

        \path[fill=mygreen!25,rounded corners=8pt] (1.3,3.45) rectangle (5.55,4.6);
        \node[rd text,anchor=west,text=mydarkgreen!50!mygreen] at (1.4,4) {\Large$\bm U$};
        \draw[edge] (2.60,3.90) -- (4.65,3.90);
        \node[vertex,label={[text=primary]above:$u$}] at (2.60,3.90) {};
        \node[vertex,label={[text=primary]above:$v$}] at (4.65,3.90) {};
        \draw[rd guide,densely dotted] (3.35,.40) -- (3.35,2.0);
        \draw[rd guide,densely dotted] (4.85,.40) -- (4.85,2.0);
        \node[rd text,anchor=north] at (3.35,.30) {$x_u$};
        \node[rd text,anchor=north] at (4.85,.30) {$x_v$};
        \draw[draw=myblue,line width=4pt] (2.0,1.80) -- (6.1,1.80);
        \draw[draw=myred,line width=4pt] (2.0,.95) -- (3.35,.95);
        \node[rd text] at (4.0,-.60) {Contribution: $1/(d+1)+x_u$};

        \begin{scope}[xshift=7.8cm]
            \path[fill=mygreen!25,rounded corners=8pt] (1.3,3.45) rectangle (3.2,4.6);
            \node[rd text,anchor=west,text=mydarkgreen!50!mygreen] at (1.4,4) {\Large$\bm U$};
            \draw[edge] (2.60,3.90) -- (4.65,3.90);
            \node[vertex,label={[text=primary]above:$u$}] at (2.60,3.90) {};
            \node[vertex,label={[text=primary]right:$v\in N_G(U)$}] at (4.65,3.90) {};
            \draw[rd guide,densely dotted] (4.85,.40) -- (4.85,2.0);
            \node[rd text,anchor=north] at (4.85,.30) {$x_v$};
            \draw[draw=myblue,line width=4pt] (2.0,1.80) -- (4.85,1.80);
            \node[rd text,font=\footnotesize,fill=white,inner sep=2pt] at (4.05,.95) {never};
            \node[rd text] at (4.2,-.60) {Contribution: $x_v$};
        \end{scope}
    \end{tikzpicture}
    \caption{How an edge contributes to the integral.}
    \label{fig:rounding-contributions}
\end{figure}
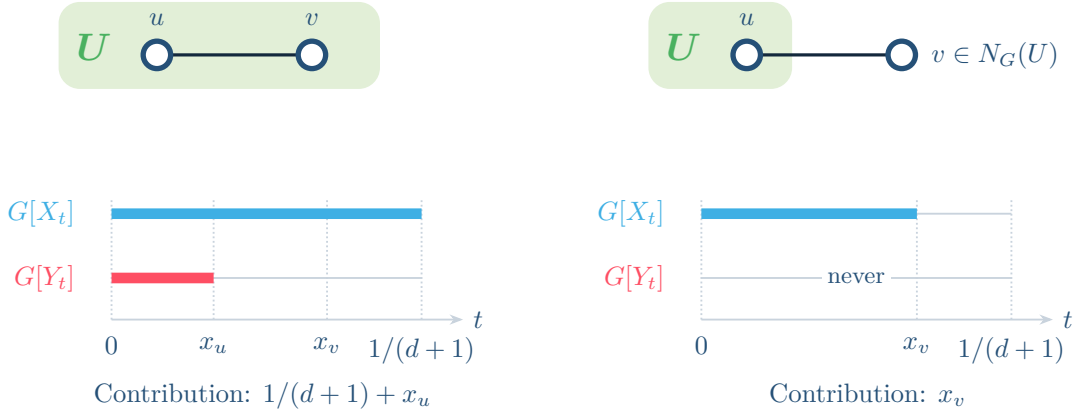

\subsection{Hard-to-dominate neighborhoods}\label{sec:auxiliary}

For the remainder of the section, fix $G$ satisfying the assumptions of \cref{thm:main}, and let $0<\rho\leq1$.
Let
\begin{equation}\label{eq:structural-parameters}
\begin{aligned}
    r&=\lceil d/2\rceil,\\
    q&=\left\lceil2\beta+2\beta/\rho\right\rceil,\\
    t&=\left\lfloor\frac{d(r+1)+2b}{2(r+1)-d}\right\rfloor+1.
\end{aligned}
\end{equation}

For the analysis, fix a minimum dominating set $D$ of $G$, so that $|D|=\MDS(G)$.
Let $T$ consist of the vertices $v\notin D$ whose open neighborhood cannot be dominated by at most $q$ vertices of $D$. 

\begin{claim}\label{claim:exceptional}
    We have $|T|\leq\rho\MDS(G)$, and $C:=D\cup T$ is a dominating set of size at most $(1+\rho)\MDS(G)$.
\end{claim}
\begin{claimproof}{claim:exceptional}
We build a depth-one shallow minor $J$ as follows. Start from $G$, keep $T$ as singleton branch sets, and assign every other vertex to a dominator in $D$ (and each vertex of $D$ to itself). Finally, contract the resulting stars to form $J$.
Since $\nabla_1(G)\leq \beta$ and since we can see $J$ as a graph on $T\cup D$, we get $|E(J)|\leq\beta|V(J)|$.
Fix $v\in T$. Each neighbor of $v$ in $J$ is a star dominated by its center vertex, or a singleton dominated by some vertex of $D$. Thus, $N_G(v)$ is dominated by at most $d_J(v)$ vertices of $D$, and $d_J(v)\geq q+1$.
By the handshake lemma: 
\[
    (q+1)|T|\leq2|E(J)|\leq2\beta(\MDS(G)+|T|),
\]
By rearranging and substituting the value of $q$, we get
\begin{align*}
    |T|&\leq\frac{2\beta\MDS(G)}{q+1-2\beta}\leq\rho\MDS(G).
\end{align*}
This finishes the proof.
\end{claimproof}

For each $v\notin D\cup T$, we can fix, by definition of $T$, some $D_v\subseteq D$ of size at most $q$ dominating $N(v)$.

\subsection{Constructing and sparsifying the hypergraph}\label{sec:hypergraph}
For some set $A\subseteq V(G)$, we write $X(A)=\bigcap_{v\in A}N(v)$. Notice that this set is disjoint from $A$.
Let $\cH_0$ be a hypergraph whose hyperedges are the nonempty sets $A$ of size at most $r$ for which no set of at most $q$ vertices outside $A$ dominates $X(A)$.
Every such hyperedge $A$ satisfies $|X(A)|>q$, since otherwise $X(A)$ itself would be an allowed dominating set. $A$'s diameter is therefore at most $2$.
The idea here is that choosing a vertex inside $A$ is much more efficient for dominating $X(A)$, because avoiding $A$ requires more than $q$ vertices to dominate $X(A)$.
By making $A$ a hyperedge, we require, as a hitting-set constraint, to select a vertex of $A$.
All these requirements can be satisfied by selecting $D\cup T$ as a hitting set, because of the following invariant:
\begin{equation}\label{eq:invariant}
    A\cap D_v\neq\varnothing \quad \forall v\in A\setminus(D\cup T).
\end{equation}
Otherwise, $D_v$ would be a dominating set of $X(A)$ outside $A$. 
In the following, we will sparsify the hypergraph $\cH_0$ so that its degree becomes bounded, while maintaining the invariant in \cref{eq:invariant}.
This invariant is stronger than just requiring $D\cup T$ to be a hitting-set, and allows us to proceed with the induction.

\paragraph{Sunflowers.}
For an integer $m\geq 2$, we say a family of distinct sets $A_1,\ldots,A_m$ is a \emph{sunflower} with \emph{kernel} $Z$ if $A_i\cap A_j=Z$ for every $i\neq j$.
The sets $A_i\setminus Z$ are called the \emph{petals}, and are all pairwise disjoint.
\Cref{fig:sunflower-definition} illustrates this definition.

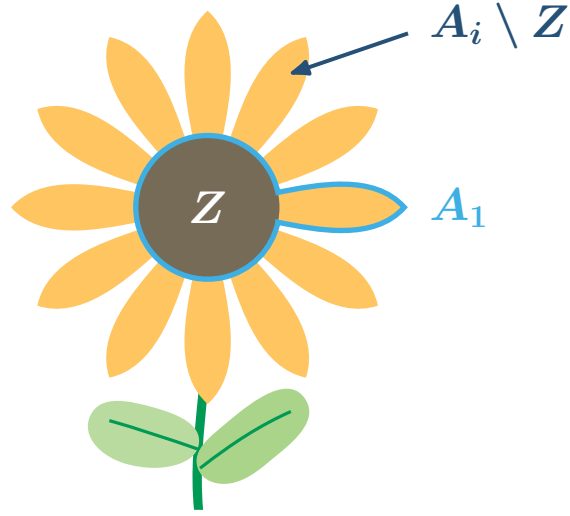
\begin{figure}[!ht]
    \centering
    \begin{tikzpicture}[x=1cm,y=1cm,use Hobby shortcut,
        every node/.style={font=\small,text=primary}]
        \begin{scope}[shift={(3.5,4.1)}]
            \path[fill=mydarkgreen]
                (-.06,-.75) .. controls (.03,-2.0) and (-.28,-3.0) .. (-.18,-4.0)
                -- (-.06,-4.0) .. controls (-.15,-2.95) and (.15,-2.0) .. (.06,-.75) -- cycle;
            \path[fill=mygreen!65]
                ([closed]-.12,-3.15) .. (-.7,-2.65) .. (-1.55,-2.75)
                .. (-1.05,-3.4) .. (-.4,-3.4);
            \path[fill=mygreen!80]
                ([closed]-.13,-3.45) .. (.25,-2.85) .. (1.3,-2.55)
                .. (.95,-3.25) .. (.3,-3.65);
            \draw[edge,draw=mydarkgreen] (-.12,-3.2) .. controls (-.55,-3.05) and (-.95,-2.92) .. (-1.3,-2.82);
            \draw[edge,draw=mydarkgreen] (-.1,-3.45) .. controls (.35,-3.1) and (.7,-2.88) .. (1.1,-2.7);

            \foreach \angle in {0,30,...,330} {
                \begin{scope}[rotate=\angle]
                    \path[fill=myorange!85]
                        ({.95*cos(12)},{-.95*sin(12)})
                        .. controls (1.45,-.30) and (2.20,-.45) .. (2.60,0)
                        .. controls (2.20,.45) and (1.45,.30) .. ({.95*cos(12)},{.95*sin(12)})
                        -- cycle;
                \end{scope}
            }
            \path[fill=myorange!35!black] (0,0) circle[radius=.95];
            \node[text=white,font=\LARGE] at (0,0) {$\bm{Z}$};

            \draw[edge,line width=2pt, draw=myblue]
                ({.95*cos(12)},{-.95*sin(12)})
                .. controls (1.45,-.30) and (2.20,-.45) .. (2.60,0)
                .. controls (2.20,.45) and (1.45,.30) .. ({.95*cos(12)},{.95*sin(12)})
                arc[start angle=12,end angle=348,radius=.95] -- cycle;
        \end{scope}

        \node[anchor=west,text=myblue] at (6.3,4.1) {\LARGE$\bm{A_1}$};
        \draw[edge,-{Latex[round]},line width=1.5pt,draw=primary] (6.15,6.4) to (4.6,5.87);
        \node[anchor=west] at (6.3,6.5) {\LARGE$\bm{A_i\setminus Z}$};
    \end{tikzpicture}
    \caption{A sunflower. The sets pairwise intersect exactly on the kernel $Z$, and the petals are pairwise disjoint.}
    \label{fig:sunflower-definition}
\end{figure}

\begin{lemma}[Sunflower lemma, Erd\H{o}s--Rado~\cite{ErdosRado1960}]\label{lem:sunflower}
    For integers $j,q\geq1$, every family of more than $j!\,q^j$ distinct sets of size $j$ contains a sunflower with $q+1$ sets.
\end{lemma}

To see why sunflowers are very relevant for us, consider a set of at most $q<m$ vertices intersecting every $A_i$.
Then, it must intersect $Z$. Otherwise, it would contain at least $m>q$ vertices, one for each petal.
Conversely, any set intersecting $Z$ intersects every $A_i$.
Thus, replacing a sunflower by its kernel preserves all hitting sets of size at most $q$.
Note that this argument is standard in the field of kernelization.

In our application, $q$ is instead a bound on the size of the sets $D_v$. We now have all the machinery necessary for sparsifying $\cH_0$.
Starting with $\cH_0$, we construct hypergraphs $\cH_1\subseteq \cH_2 \subseteq \ldots \subseteq \cH_{r-1},\cH_r$.
For each $j=1,\ldots,r-1$, we get $\cH_j$ from $\cH_{j-1}$ by adding the nonempty kernel of every sunflower of $q+1$ distinct hyperedges of size $r-j+1$ in $\cH_{j-1}$.
These stages deal with hyperedges of sizes $r,r-1,\ldots,2$ in this order. Thus, every new hyperedge is smaller than the sets used to construct it.
Finally, get $\cH_r$ from $\cH_{r-1}$ by keeping only its inclusion-minimal hyperedges and deleting all isolated vertices. Set $\cH=\cH_r$.
Let us prove by induction that every new kernel $Z$ satisfies~\cref{eq:invariant}. Let the corresponding sunflower be $A_1,\ldots,A_m$.
Take any $v\in Z\setminus(D\cup T)$. Since $v$ is in every $A_i$, the invariant for these hyperedges gives $A_i\cap D_v\neq\varnothing$ for every $i$.
As $|D_v|\leq q$, the previous argument implies that $D_v$ intersects $Z$, as wanted.
Thus, every hyperedge in the reduction satisfies~\cref{eq:invariant}, as the base case is already satisfied on $\cH_0$ by definition.
Moreover, as $D_v\subseteq D$, every hyperedge considered here intersects $D\cup T$.

\begin{claim}\label{claim:sunflower}
$C\cap V(\cH)$ is a hitting set of $\cH$, and every hyperedge of $\cH_0$ contains a hyperedge of $\cH$.
Every hyperedge of $\cH$ has diameter at most two in $G$.
The hypergraph $\cH$ has rank at most $r$ and maximum degree at most $B_1:=\sum_{j=0}^{r-1}j!\,q^j$.
Moreover,
\begin{align*}
    |V(\cH)|\leq rB_1(1+\rho)\MDS(G).
\end{align*}
\end{claim}
\begin{claimproof}{claim:sunflower}
First, notice that \cref{eq:invariant} already proves that $C\cap V(\cH)$ is a hitting set of $\cH$, as $D_v\subseteq D$.
Every hyperedge of $\cH_0$ is in $\cH_{r-1}$, and contains an inclusion-minimal hyperedge of $\cH_{r-1}$, which is in $\cH$.
Moreover, the hyperedges of $\cH_0$ have size at most $r$ and diameter at most two in $G$.
By construction, this is still true for $\cH_{r-1}$, and hence for $\cH$.
To prove the degree bound, fix some vertex $v\in V(\cH)$.
There is at most one singleton hyperedge containing $v$.
For $1\leq j\leq r-1$, more than $j!\,q^j$ hyperedges of size $j+1$ containing $v$ would, after deleting $v$, contain a sunflower with $q+1$ sets by \cref{lem:sunflower}.
Adding $v$ back proves that a nonempty proper kernel was added at some point during the reduction, contradicting the minimality of $\cH$. This proves the bound.
Finally, by \cref{claim:exceptional} and because $C\cap V(\cH)$ is a hitting set of $\cH$, we get:
\[
    |V(\cH)|\leq r|E(\cH)|\leq rB_1|C|\leq rB_1(1+\rho)\MDS(G).\qedhere
\]
\end{claimproof}

\subsection{Bounding the size of the residual set}\label{sec:residual}

We first need the following standard result.
\begin{claim}\label{claim:bipartite-density}
    $K_{r+1,t}$ is not a subgraph of $G$.
    Moreover, $r\leq d+1$.
\end{claim}
\begin{proof}
Since $r\geq d/2$, we have $2(r+1)-d\geq2>0$.
The definition of $t$ gives $\bigl(2(r+1)-d\bigr)t>d(r+1)+2b$ and $(r+1)t>\frac d2(r+1+t)+b$.
$K_{r+1,t}$ has $(r+1)t$ edges and $r+1+t$ vertices.
If $K_{r+1,t}$ was a subgraph of $G$, this would contradict the density assumption of \cref{thm:main}.
Finally, $r=\lceil d/2\rceil\leq d/2+1\leq d+1$.
\end{proof}

Define $X=\bigcup_{e\in E(\cH)}X(e)$, so that every hitting set of $\cH$ dominates $X$. Also define $R_0=V(G)\setminus X$, the rest of the vertices not in $X$, that may be undominated by a hitting set of $\cH$.

\begin{claim}\label{claim:residual-degree}
Every vertex has a constant number of neighbors in $R_0$:
\begin{align*}
    |N_G(v)\cap R_0|\leq B_2:=q^{r-1}(qt+r)\quad\forall v\in V(G).
\end{align*}
In particular, $R_0$ has size at most $(B_2+1)\MDS(G)$.
\end{claim}
\begin{claimproof}{claim:residual-degree}
To prove this bound, we first need a preliminary observation. If $X(A)$ intersects $R_0$ and $|A|\leq r$, then $A\notin E(\cH_0)$: otherwise $A$ would contain some $e\in E(\cH)$, and so $X(A)\subseteq X(e)\subseteq X$, a contradiction. Therefore, by the definition of $\cH_0$, there is a set of at most $q$ vertices outside $A$ dominating $X(A)$.

Now suppose for the sake of contradiction that some vertex $v$ violates the bound in~\cref{claim:residual-degree}, and set $A_1=\{v\}$, so that $X(A_1) = N_G(v)$.
We inductively build the sets $A_2, \dots, A_r$ by adding well-chosen vertices, such that $|A_i|=i$ and such that each $X(A_i)$ intersects $R_0$.
This trivially holds for $A_1$, as $|N_G(v)\cap R_0|>B_2$.
Suppose the sets $A_1,\dots,A_i$ have been constructed, where $1\leq i<r$.
By the previous observation, there is a set of at most $q$ vertices outside $A_i$ dominating $X(A_i)$.
By the pigeonhole principle, it contains a vertex $z$ dominating at least $|X(A_i)\cap R_0|/q$ vertices of $X(A_i)\cap R_0$.
Choose such a $z$ and set $A_{i+1}=A_i\cup\{z\}$. We have:
\[
    |X(A_{i+1})\cap R_0| =|X(A_i)\cap R_0\cap N_G(z)|\geq |X(A_i)\cap R_0\cap N_G[z]|-1\geq \frac{|X(A_i)\cap R_0|}{q}-1.
\]
By choice of $B_2$, we get that $|X(A_r)\cap R_0|>qt$.
But each $z\notin A_r$ dominates at most $t$ vertices of $X(A_r)$, as otherwise $K_{r+1,t}$ would be a subgraph of $G$ (contradicting \cref{claim:bipartite-density}), with $A_r\cup\{z\}$ on one side and $t$ vertices of $X(A_r)$ on the other. 
Thus, $q$ vertices outside $A_r$ can only dominate at most $qt$ vertices of $X(A_r)$, i.e. not its entirety.
This contradicts the existence of a set of at most $q$ vertices outside $A_r$ dominating $X(A_r)$.
Finally, since $D$ dominates $R_0$, the degree bound gives $|R_0|\leq(B_2+1)\MDS(G)$.
\end{claimproof}

\subsection{Preparing the LP}\label{sec:lp-preparation}
Consider the following LP:
\begin{align*}
    \min\sum_v x_v\quad\text{subject to}\quad
    \left\{
    \begin{aligned}
        x&\geq0,\\
        \sum_{v\in e}x_v&\geq1 &&\forall e\in E(\cH),\\
        ~~\sum_{v\in N[w]}x_v&\geq1 &&\forall w\in R_0.
    \end{aligned}
    \right.
\end{align*}
It accomplishes two things: first, we want to hit every hyperedge of $\cH$. In particular, this will fractionally dominate every vertex in $X$. Secondly, we ask for vertices not in $X$, that are maybe not dominated by the hitting set requirement, to be fractionally dominated.
An integral solution to this gives a dominating set for our whole graph.
This is a stronger requirement than a usual dominating set, but it is a very good approximation of the original problem. 
Indeed, as $C$ contains a dominating set and is a hitting set of $\cH$, the integral optimum of this LP is at most $|C|\leq(1+\rho)\MDS(G)$. Thus, when $\eps$ goes to zero, so does $\rho$, and the integral optimum is arbitrarily close to the MDS size.

We use the following neighborhood-complexity bound.

\begin{proposition}[Gajarsk{\'y} et al.~{\cite[Lemma~4.3(2)]{GajarskyEtAl2017}}]\label{prop:neighborhood-complexity}
For every $\beta\geq1$, there is a constant $c_\beta$ such that, for every graph $G$ with $\nabla_1(G)\leq\beta$ and every $R\subseteq V(G)$, the vertices $v\in V(G)\setminus R$ have at most $c_\beta|R|$ distinct nonempty neighborhoods $N_G(v)\cap R$.
\end{proposition}

To approximate this LP locally, we work on a slightly smaller graph than $G$.
Note that this idea is an adaptation of Heydt et al.~\cite[Sections~6.3--6.4]{HeydtEtAl2025}.
We modify it slightly to keep every vertex of $V(\cH)$, so that the additional hitting constraints survive.
Keep all vertices of $V(\cH)\cup R_0$, remove the other vertices with no neighbors in $R_0$, and only keep one vertex (called the representative) for each distinct nonempty $R_0$-neighborhood of the vertices outside $V(\cH)\cup R_0$ (say, the one with the smallest identifier).
Let $G'$ be the induced graph on those remaining vertices.
Define the set $C'\subseteq V(G')$ as follows. It will only be used for the analysis. 
Start with the vertices of $C$ in $G'$ and also add, for every vertex $v\in C\setminus V(G')$ with a neighbor in $R_0$, its $R_0$-representative chosen when constructing $G'$.

\begin{claim}\label{claim:compression}
The set $C'$ is a solution to the LP on $G'$. Moreover, we have $|C'|\leq|C|$, and
\begin{align*}
    |V(G')|\leq\bigl(rB_1(1+\rho)+(1+c_\beta)(B_2+1)\bigr)\MDS(G).
\end{align*}
\end{claim}
\begin{claimproof}{claim:compression}
$C'$ is still a hitting-set of $\cH$ in $G'$, since it contains all vertices in $C\cap V(\cH)$.
It also dominates $R_0$ in $G'$, since it contains all vertices in $C\cap R_0$, and each vertex in $C\setminus C'$ either has no neighbor in $R_0$ or has been replaced by a vertex with the same $R_0$-neighborhood.
Thus $C'$ is a solution for the LP on $G'$, and $|C'|\leq |C\cap V(G')|+|C\setminus V(G')|=|C|$.
\Cref{prop:neighborhood-complexity} and \cref{claim:exceptional,claim:sunflower,claim:residual-degree} give
\[
    |V(G')|\leq\bigl(rB_1(1+\rho)+(1+c_\beta)(B_2+1)\bigr)\MDS(G).\qedhere
\]
\end{claimproof}
Define
\begin{equation}\label{eq:cutoff}
\begin{aligned}
    \Gamma&=\left\lceil\frac{2\beta}{\rho}\bigl(rB_1(1+\rho)+(1+c_\beta)(B_2+1)\bigr)\right\rceil,\\
    P&=\{v\in V(G')\mid d_{G'}(v)>\Gamma\}.
\end{aligned}
\end{equation}
Let $F=G'-P$, $\cH'=\{e\in E(\cH)\mid e\cap P=\varnothing\}$, and $R_1=R_0\setminus N_G[P]$.
$\cH'$ may be viewed as having all vertices of $F$ as its ground set, at the expense of adding isolated vertices.
\Cref{fig:algorithm-sets} illustrates these sets and their relations.
\begin{figure}[H]
    \centering
    \begin{tikzpicture}[x=1cm,y=1cm,use Hobby shortcut,
        every node/.style={font=\small,text=primary}]
        \def\algorithmResidualRegion{
            ([closed]8.55,6.55) .. (10.7,7.65) .. (13.0,6.95)
            .. (13.65,4.8) .. (12.95,2.25) .. (10.85,1.85)
            .. (8.8,2.55) .. (8.35,4.6)}
        \def\algorithmClosedNeighborhood{
            ([closed]3.55,5.4) .. (4.25,7.15) .. (7.1,7.65)
            .. (9.65,7.0) .. (10.6,5.1) .. (10.15,3.0)
            .. (8.15,1.75) .. (5.2,1.95) .. (3.65,3.25)}

        \path[fill=myblue!20]
            ([closed].25,4.6) .. (.8,8.5) .. (6.8,8.9)
            .. (13.8,8.3) .. (14.3,4.3) .. (13.5,.65)
            .. (7,.25) .. (.9,.75);
        \node[anchor=west] at (2,8.35) {\Large $\bm{V(G')}$};

        \path[fill=myblue!50] \algorithmResidualRegion;
        \begin{scope}
            \clip \algorithmResidualRegion;
            \path[fill=myblue,even odd rule]
                (-1,-1) rectangle (15,10) \algorithmClosedNeighborhood;
        \end{scope}

        \path[fill=mygreen]
            ([closed]2.25,5.9) .. (3.5,7.0) .. (4.9,6.65)
            .. (5.6,5.2) .. (5.5,3.25) .. (4.2,1.6)
            .. (2.45,1.95) .. (1.95,3.9);
        \node[text=mydarkgreen,font=\large] at (2.65,4.5) {\Large$\bm C'$};

        \path[fill=myorange]
            ([closed]5.0,4.75) .. (5.65,5.85) .. (7.15,6.0)
            .. (8.0,5.05) .. (7.7,3.85) .. (6.1,3.65)
            .. (5.15,4.0);
        \node[text=Orange!60!black,font=\large] at (7.1,4.8) {\Large$\bm P$};

        \node[text=myorange!50!myred,anchor=south] at (7.1,7.8) {\Large$\bm{N_{G'}[P]}$};
        \node at (9.35,5.95) {$R_0\setminus R_1$};
        \node[font=\large] at (12.05,4.6) {\Large$\bm R_1$};

        \draw[edge,line width=2pt,draw=myred,dashed]
            ([closed]4.35,5.3) .. (4.85,5.95) .. (5.9,5.75)
            .. (6.2,5.05) .. (5.2,4.95);
        \draw[edge,line width=2pt,draw=myred,dashed]
            ([closed]4.5,3.75) .. (5.15,4.25) .. (6.2,4.1)
            .. (6.45,3.55) .. (5.4,3.25);
        \begin{scope}[xshift=-0.8cm,yshift=-0.8cm]
            \draw[edge,line width=2pt,draw=mypurple!50]
                ([closed]1.75,2.95) .. (2.4,3.55) .. (3.6,3.35)
                .. (3.9,2.7) .. (2.6,2.4);
        \end{scope}
        \draw[edge,line width=2pt, draw=mypurple!50]
            ([closed]2.0,6.1) .. (2.75,6.75) .. (3.85,6.45)
            .. (4.05,5.8) .. (2.9,5.6);

        \draw[edge,line width=2pt,draw=myorange!50!myred,fill=none]
            \algorithmClosedNeighborhood;

        \draw[edge,-{Stealth[length=2.5mm]},draw=mygreen]
            (4.2,1.6) .. controls (6.0,.7) and (9.4,.7) .. (10.3,1.9);
        \node[text=mygreen!50!ForestGreen,font=\footnotesize] at (7.5,.65) {dominates};
        
        \draw[edge,draw=myred,dashed]
            ([closed]1.0,-.55) .. (1.35,-.25) .. (1.7,-.55) .. (1.35,-.8);
        \node[anchor=west] at (1.9,-.55) {$\cH\setminus\cH'$: hit by $P$};
        \draw[edge,draw=mypurple!50]
            ([closed]8.0,-.55) .. (8.35,-.25) .. (8.7,-.55) .. (8.35,-.8);
        \node[anchor=west] at (8.9,-.55) {$\cH'$: disjoint from $P$};
    \end{tikzpicture}
    \caption{The orange curve represents the border of $N_{G'}[P]$.
    Inside it, $P$ dominates $R_0\setminus R_1$. The darker region $R_1$ is outside of $N_{G'}[P]$ and is not dominated by $P$.
    The set $C'\setminus P$ is a hitting set of $\cH'$ and dominates $R_1$.}
    \label{fig:algorithm-sets}
\end{figure}
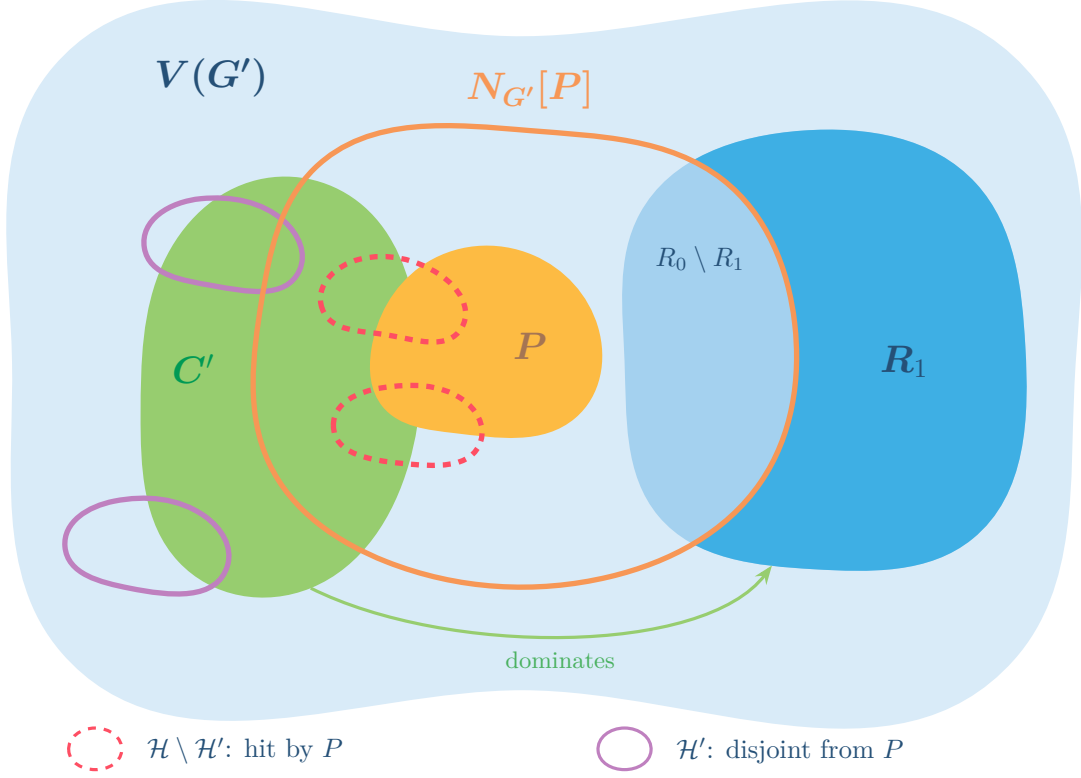

Let us translate our results from before to the graph $G'$.

\begin{claim}\label{claim:cutoff}
    The following hold:
    \begin{itemize}
        \item $P$ has size at most $\rho\MDS(G)$,
        \item $F$ has maximum degree at most $\Gamma$ and satisfies the density assumption of \cref{thm:main},
        \item $C'\setminus P$ is a hitting set of $\cH'$, dominates $R_1$ in $F$, and has size at most $(1+\rho)\MDS(G)$,
        \item $\cH'$ has rank at most $r\leq d+1$ and maximum degree at most $B_1$.
    \end{itemize}
\end{claim}
\begin{claimproof}{claim:cutoff}
As $G'$ is a depth-one shallow minor, we have $|E(G')|\leq\beta|V(G')|$.
The handshake lemma, \cref{claim:compression}, and \cref{eq:cutoff} together give
\begin{align*}
    |P|\leq\frac{2\beta|V(G')|}{\Gamma}\leq\rho\MDS(G).
\end{align*}
Every vertex of $F$ in $G'$ has degree at most $\Gamma$, so $F$ has maximum degree at most $\Gamma$.
Every hyperedge of $\cH'$ is disjoint from $P$, so it is intersected by $C'\setminus P$.
Similarly, every vertex of $R_1$ has no vertex in $P$ inside its closed neighborhood, so it is dominated by $C'\setminus P$ in $F$.
The $(1+\rho)\MDS(G)$ bound follows from \cref{claim:exceptional,claim:compression}.
As said before, $F$ is a subgraph of $G'$, so the density bound holds for all its subgraphs.
The bounds of the rank and degree of $\cH'$ follow from \cref{claim:bipartite-density,claim:sunflower}.
\end{claimproof}

\begin{claim}\label{claim:recovery}
    If $S\subseteq V(F)$ is a hitting set of $\cH'$ and dominates $R_1$ in $F$, then $P\cup S$ dominates $G$.
\end{claim}
\begin{claimproof}{claim:recovery}
Every hyperedge of $\cH$ either intersects $P$, or is in $\cH'$ and thus intersects $S$.
Every vertex in some $X(e)$ for $e\in E(\cH)$ is either dominated by some vertex of $P$ or some vertex of $S$.
Thus, $P\cup S$ dominates $X=\bigcup_{e\in E(\cH)}X(e)$.
All vertices not in $X$ are in $R_0$ by definition of $R_0$.
$P$ dominates $R_0\setminus R_1$ by definition of $R_1$, and $S$ dominates $R_1$.
Thus, all vertices are dominated by $P\cup S$.
\end{claimproof}

Consider the LP with variables on $V(F)$, hitting constraints on $E(\mathcal H')$, and domination constraints on $N_F[w]$ for $w\in R_1$.

\begin{claim}\label{claim:lp-incidence}
    The row and column sums of the constraint matrix are bounded by $\max\{r,\Gamma+1\}$ and $B_1+\Gamma+1$, respectively.
\end{claim}
\begin{claimproof}{claim:lp-incidence}
First, notice that $F$ has maximum degree at most $\Gamma$ by the definition of $P$.
A hitting constraint contains at most $r$ variables, while a domination constraint for $w\in R_1$ contains the at most $\Gamma+1$ variables of $N_F[w]$.
Each variable $x_v$ appears in at most $B_1$ hitting constraints, since $\cH'$ has maximum degree at most $B_1$, and is in at most $|N_F[v]\cap R_1|\leq\Gamma+1$ domination constraints.
\end{claimproof}

\subsection{Using the LP algorithm}\label{sec:locality}

Given an $m\times n$ $0$--$1$ matrix $A$, define its incidence network as the bipartite graph with bipartition sizes $m$ and $n$.
Each of the $m$ constraints is mapped to one vertex of the first partition, and each of the $n$ variables is mapped to one vertex of the second partition. Both mappings are bijective.
Constraint $i$ and variable $j$ are adjacent in the bipartite graph if and only if $A_{ij}=1$.

Let us denote the vector that contains only ones as $\mathbf{1}$.
We use the following LP algorithm, also used by Heydt et al.~\cite[Section~6.2]{HeydtEtAl2025}.
\begin{proposition}[Kuhn, Moscibroda, and Wattenhofer~{\cite[Corollary~4.1]{KuhnEtAl2006}}]\label{prop:covering-lp}
Let $A\in\{0,1\}^{m\times n}$ be a matrix with no zero row, with row sums at most $M\geq2$ and column sums at most $Q\geq2$. Interpret $m$ as the number of constraints and $n$ as the number of variables.
For every $0<\rho\leq1$, given $M$ and $Q$, there exists a deterministic algorithm computing a $(1+\rho)$-approximate solution on the incidence network of $A$ to $\min\{\sum_{j=1}^n x_j\mid Ax\geq\mathbf{1},\ x\geq0\}$ in $O(\rho^{-4}\log M\log Q)$ rounds in the \LOCAL model.
\end{proposition}

In this subsection, we check that the LP incidence network and machinery can be implemented in the original graph.

\begin{lemma}\label{lem:locality}
The hypergraphs $\cH,\cH'$, the induced graphs $G',F$, and the sets $P,R_0,R_1$ can be computed in a constant number of rounds, independent of $\rho$.
This requires only polynomial local computation for fixed $d,\beta,b,\rho$.
One round on the incidence network of the residual LP can be simulated in at most two rounds on $G$.
\end{lemma}
\begin{proof}
A vertex $v$ determines its incident $\cH_0$-hyperedges by only looking at distance 3.
Indeed, a candidate hyperedge $A$ containing $v$ has diameter at most two whenever its common neighborhood is nonempty.
Therefore, any vertex useful for dominating $X(A)$ is at distance at most two from $v$.
This means that $v$ can check all candidate hyperedges and all outside dominating sets of size at most $q$ from this view.
If $v$ is in some hyperedge $A$ of $\cH_j$, then it also is in every hyperedge of $\cH_i$ ($i<j$) used to build $A$.
Thus, $v$ can compute all its incident hyperedges in $\cH_1,\dots, \cH_{r-1}$ in only three rounds.
Vertices can send the list of hyperedges they belong to to all their distance-two neighbors. With this information, they remove hyperedges that are not minimal.
This creates $V(\cH)$ in constant rounds.
In constant additional rounds, the vertices also compute the sets $X$ and $R_0$.

Vertices with the same nonempty $R_0$-neighborhood are at distance at most two. They can thus agree on a representative, say the vertex with the smallest identifier, in constant rounds.
Therefore, $G'$ is computed in constant rounds. Then, each vertex computes its degree in $G'$ to decide its membership to $P$.
Two more rounds are needed to compute the hyperedge and domination constraints satisfied by $P$. With this information, the sets $F$, $\cH'$, and $R_1$ are computed.

Each vertex of $F$ simulates its own LP variable.
Every hyperedge chooses a vertex (say its smallest identifier vertex) which simulates its constraint, and each $w\in R_1$ simulates its domination constraint.
By \cref{claim:sunflower}, the vertices that simulate constraints are within distance two in $G$ of the vertices that simulate their variables.
Note that vertices outside $F$ can still transmit messages. Therefore, in total, one round on the incidence network costs at most two rounds on the original graph.

Let us now prove that the computations are polynomial in the size of that view.
Each vertex checks only sets of at most $r$ vertices, outside dominating sets of size at most $q$, and families of $q+1$ hyperedges.
Furthermore, there are polynomially many such choices in its view.
There are only $r-1$ reductions, and all kernels have size at most $r$, thus the computation of $\cH$ is polynomial.
The other operations are also polynomial, proving the wanted result.
\end{proof}

We can now prove the following.

\begin{lemma}\label{lem:local-lp}
For every $0<\rho\leq1$, the residual LP on $F$ with hyperedges $\cH'$ and domination requests $R_1$ has a fractional solution $x$ with
\begin{align*}
    \sum_{v\in V(F)}x_v\leq(1+\rho)^2\MDS(G).
\end{align*}
Moreover, one can compute a solution in $O_{d,\beta,b}(\rho^{-4}\log(1/\rho)^2)$ \LOCAL rounds, with polynomial local computation for fixed $d,\beta,b,\rho$.
\end{lemma}
\begin{proof}
By \cref{claim:cutoff}, the set $C'\setminus P$ is a solution to the LP and has size at most $(1+\rho)\MDS(G)$.
It thus suffices to compute a $(1+\rho)$-approximation of the residual LP to get the inequality in the statement.
If there are no constraints, use $x=0$.
Otherwise, there cannot be any zero row.
By \cref{claim:lp-incidence}, we can apply \cref{prop:covering-lp} with $M=\max\{2,r,\Gamma+1\}$ and $Q=\max\{2,B_1+\Gamma+1\}$.
We also have 
\begin{align*}
    q&=O_\beta(\rho^{-1}),\\
    B_1&=O_{d,\beta}(\rho^{-(r-1)}),\\
    B_2&=O_{d,\beta,b}(\rho^{-r}),\\
    \Gamma&=O_{d,\beta,b}(\rho^{-(r+1)}).
\end{align*}
Thus, the LP algorithm of Kuhn, Moscibroda, and Wattenhofer~{\cite[Corollary~4.1]{KuhnEtAl2006}} runs in $O_{d,\beta,b}(\rho^{-4}\log(1/\rho)^2)$ rounds in the incidence network.
By \cref{lem:locality}, one round on the incidence network of the residual LP can be simulated in at most two rounds on $G$. Thus, we get the same asymptotic running time in $G$.

Note that the LP algorithm only does computations on the (bounded number of) entries in each row and column.
All calculations can be done exactly with rational numbers, by choosing $\rho$ rational and rounding the LP algorithm's bounds up to suitable powers of $1+\rho/10$.
As this is only a technical detail, we refer the reader to the original analysis of the LP algorithm~\cite[Algorithm~1 and Theorem~4.1]{KuhnEtAl2006}.
The running time is polynomial in the neighborhood size, but its exponent may depend on $d,\beta,b,\rho$.
\end{proof}

\subsection{Taking care of the additive bound in the density}\label{sec:additive}

We now put together our previous algorithms. We get a returned dominating set whose size is a constant times the optimal, plus a small additive term.
\begin{proposition}\label{prop:additive-algorithm}
For every $0<\rho\leq1$, there is a deterministic \LOCAL algorithm that, under the assumptions of \cref{thm:main}, computes a dominating set of size at most
\begin{align*}
    \bigl(\rho+(d+1)(1+\rho)^2\bigr)\MDS(G)+\frac{4b}{d},
\end{align*}
in $O_{d,\beta,b}(\rho^{-4}\log(1/\rho)^2)$ rounds. It only requires polynomial local computation for fixed $d,\beta,b,\rho$.
\end{proposition}
\begin{proof}
Compute the solution of \cref{lem:local-lp}.
We can use the rounding algorithm from \cref{lem:rounding} on $F$ with requests $R_1$ and hyperedges $\cH'$.
\Cref{claim:cutoff} makes sure that all conditions necessary to apply this algorithm are met.
The rounding algorithm returns a set $S$ which dominates $R_1$ and hits $\cH'$. We also have the bound:
\begin{align*}
    |S|\leq(d+1)\sum_{v\in V(F)}x_v+\frac{4b}{d}.
\end{align*}
By \cref{claim:recovery}, $P\cup S$ dominates $G$.
$P$ has size at most $\rho \MDS(G)$ by \cref{claim:cutoff}.
This together with \cref{lem:local-lp} proves the bound on the size of the returned dominating set.
The complexity follows from \cref{lem:local-lp,lem:rounding}.
\end{proof}

The next lemma is the reduction used by Bonamy et al.~\cite[Proposition~D.2]{BonamyEtAl2026}.
We give some more detail to explain why the local computation is polynomial.

\begin{lemma}\label{lem:additive-removal}
Let $\cC$ be a graph class and let $\delta>0$.
Suppose some deterministic $T$-round \LOCAL algorithm returns a dominating set of size at most $\alpha\MDS(G)+\lambda$ on every $G\in\cC$, for some $\alpha\geq1$ and $\lambda\geq0$.
Then it can be turned into an $(\alpha+\delta)$-approximation on $\cC$ that runs in $T+O(1+\lambda/\delta)$ rounds.
If the original algorithm has polynomial local computation, so does the new one for fixed $\lambda,\delta$.
\end{lemma}
\begin{proof}
Let $h=\max\{1,\lceil\lambda/\delta\rceil\}$.
If $\MDS(G)\leq h$, then $G$ has diameter at most $3h-1$.
Indeed, consider a shortest path. Every vertex of $G$ can only dominate at most $3$ vertices of this path. 
Therefore, a dominating set of size $h$ can only dominate $3h$ vertices of this path.

The algorithm goes as follows.
First, every vertex computes its radius-$(3h+1)$ neighborhood.
If the whole graph is visible, it checks if any vertex set of size at most $h$ dominates the whole graph.
If it is the case, it chooses a canonical minimum one (for instance, the smallest one in the lexicographic order given by the vertex identifiers).
If the whole graph is not visible, or if no vertex set of size at most $h$ dominates the graph, it uses the original algorithm.

The important point is that all vertices agree on this choice.
When $\MDS(G)\leq h$, every vertex sees the entire graph and chooses the same solution.
When $\MDS(G)>h$, no vertex finds a dominating set of size at most $h$, thus they all use the original algorithm.
In the latter case, by definition of $h$, we have $\lambda\leq\delta\MDS(G)$. Thus the $\alpha+\delta$ approximation ratio is proved.
This algorithm uses only $O(h)$ additional rounds, and enumerating sets of size at most $h$ is polynomial for fixed $\lambda,\delta$.
\end{proof}

\subsection{Putting the algorithm together}\label{sec:assembly}

\begin{proof}[Proof of \Cref{thm:main}]
Put
\begin{align*}
    \rho=\min\left\{1,\frac{\eps}{3(d+1)+2}\right\}.
\end{align*}
Use the algorithm of \cref{prop:additive-algorithm}.
Apply \cref{lem:additive-removal} with
\begin{align*}
    \alpha&=\rho+(d+1)(1+\rho)^2,\\
    \lambda&=4b/d,\\
    \delta&=\rho.
\end{align*}
We get an approximation factor of $\alpha+\delta=2\rho+(d+1)(1+\rho)^2$.
We can also bound $(1+\rho)^2\leq1+3\rho$ because $0<\rho\leq1$.
Therefore we get:
\begin{align*}
    2\rho+(d+1)(1+\rho)^2
    \leq d+1+\bigl(3(d+1)+2\bigr)\rho
    \leq d+1+\eps.
\end{align*}
For sufficiently small $\eps$, we have $\rho=\eps/(3(d+1)+2)=\Theta_d(\eps)$, and
\begin{align*}
    \rho^{-4}\log(1/\rho)^2+1+\frac{4b}{d\rho}
    =O_{d,b}\bigl(\eps^{-4}\log(1/\eps)^2\bigr).
\end{align*}
By using this bound, and \cref{prop:additive-algorithm,lem:additive-removal}, we get the stated round complexity.
Polynomial local computation follows from \cref{prop:additive-algorithm,lem:additive-removal}.
\end{proof}

Applying \cref{thm:main} with $b=0$ gives \cref{cor:mad}, which we restate here.
\madcorollary*

\subsection{Tightness}
Amiri and Siebertz~\cite[Theorem~8.7]{AmiriSiebertz2026} show that, for every integer $r\geq1$ and every $\delta>0$, no deterministic constant-round \LOCAL algorithm achieves a $2r+1-\delta$ approximation ratio on the $r$-th powers $P_n^r$ of paths.
For even $d=2r$, the hereditary closure of these path powers has maximum degree at most $d$, hence bounded expansion and maximum average degree at most $d$.
The same power-of-paths graphs were also used before in the author's thesis~\cite[Lemma~V.2.1]{Picavet2026Thesis}.
They have pathwidth at most $r$, hence $\nabla_1\leq r$, proving the lower bound.
Thus the factor $d+1$ is sharp as an infimum for every even integer $d\geq2$.
Note that, additionally, this lower bound answers negatively a question of Lenzen and Wenning~\cite{LenzenWenning2024}. Without their high-girth assumption, there is no constant round algorithm with approximation ratio sublinear in $\nabla_1$.

\section{Applications}

Our algorithm can be applied to many graph classes. In the following, we use it to give almost-tight results on a few important graph classes.

\paragraph{Bounded Euler genus.}
On graphs of bounded Euler genus, the maximum average degree can be larger than $6$ when the genus is non-zero, but only because of an additional additive term.
Thus, we can directly apply \cref{thm:main}.
\begin{corollary}\label{cor:genus}
    For every integer $g$ and every $\eps>0$, there is a deterministic constant-round \LOCAL algorithm that computes a $(7+\eps)$-approximation for \MinDS on graphs of Euler genus at most $g$.
    For $\eps>0$, the algorithm runs in $O_g(\eps^{-4}\log(1/\eps)^2)$ rounds, with polynomial local computation for fixed $g$ and $\eps$.
    No embedding is required.
\end{corollary}
Note that this also applies to planar graphs. Thus, the coefficient $7+\eps$ is optimal up to the $\eps$ term.
\begin{proof}
By Euler's inequality on graphs of Euler genus at most $g$, we have $|E(H)|\leq3|V(H)|+3g$ for every minor $H$ of $G$.
This also applies to every depth-one shallow minor of $G$.
Apply \cref{thm:main} with $d=6$, $b=3g$ and $\beta=3+3g$.
\end{proof}

\paragraph{Triangle-free graphs of bounded Euler genus.}
\begin{corollary}\label{cor:triangle-free-genus}
    For every integer $g$ and every $\eps>0$, there is a deterministic constant-round \LOCAL algorithm that computes a $(5+\eps)$-approximation for \MinDS on triangle-free graphs of Euler genus at most $g$.
    For $\eps>0$, the algorithm runs in $O_g(\eps^{-4}\log(1/\eps)^2)$ rounds, with polynomial local computation for fixed $g$ and $\eps$.
    No embedding is required.
\end{corollary}
Note that this also applies to triangle-free planar graphs. Thus, the coefficient $5+\eps$ is optimal up to the $\eps$ term~\cite[Section~4]{BonamyEtAl2021Outerplanar}.
\begin{proof}
By Euler's inequality on triangle-free graphs of Euler genus at most $g$, we have $|E(H)|\leq 2|V(H)|+2g$ for every subgraph $H$ of $G$.
The genus does not increase by taking minors, therefore every depth-one shallow minor has Euler genus at most $g$ and density at most $3+3g$, by Euler's inequality on general graphs of Euler genus at most $g$.
Apply \cref{thm:main} with $d=4$, $b=2g$, and $\beta=3+3g$.
\end{proof}

The same approximation ratio can also be achieved by using the recent meta-theorem of Wawrzyniak~\cite[Theorem~6.3 and Sections~7.1--7.3]{Wawrzyniak2026Transfer} on a $(5+\eps)$-approximation algorithm for triangle-free planar graphs.

\paragraph{Further tight bounds.}
The \emph{queue-number} (respectively, \emph{stack-number}) of a graph is the smallest integer $k$ such that its vertices can be linearly ordered and its edges partitioned into $k$ sets, with the property that no two nested (respectively, crossing) edges are in the same set.
Here, given a linear order $<$ on the vertices, edges $uv$ and $xy$ are \emph{nested} when $u<x<y<v$, and \emph{cross} when $u<x<v<y$.
The bounded treewidth, bounded pathwidth, and $K_t$-minor-free classes are minor-closed, but not bounded queue-number and stack-number classes.
Indeed, for $k\geq2$ in the queue-number case and $k\geq3$ in the stack-number case, these classes contain subdivisions of any graph and thus are not minor-closed. They are also not contained in any proper minor-closed class~\cite{NesetrilEtAl2012}.

\begin{corollary}\label{cor:other-classes}
    For a fixed integer $k\geq1$, and every $\eps>0$, there exist deterministic \LOCAL algorithms that achieve the following approximation ratios on the following graph classes:
    \begin{center}
    \renewcommand{\arraystretch}{1.15}
    \begin{tabular}{@{}l|c@{}}
        \hline
        \text{Graph class} & \text{Approximation ratio}\\[2pt]
        \hline
        $\text{Queue-number at most }k$ & $4k+1+\eps$\\
        $\text{Stack-number at most }k$ & $2k+3+\eps$\\
        $\text{Treewidth or pathwidth at most }k$ & $2k+1+\eps$\\
        $K_k\text{-minor-free graphs}$ & $2k-3+\eps\text{ for }3\leq k\leq 9$\\
        \hline
    \end{tabular}
    \end{center}
    The algorithms run in $O_k(\eps^{-4}\log(1/\eps)^2)$ rounds, with polynomial local computation for fixed $k,\eps$.
    No layout or decomposition is required as input.
    All these algorithms are tight up to the $\eps$ term.
\end{corollary}
\begin{proof}
All graphs of bounded queue- or stack-number have bounded expansion~\cite{NesetrilEtAl2012}.
Their maximum average degrees are at most $4k$ and $2k+2$, respectively.
Indeed, a non-nested edge set has at most $2n-3$ edges, and a graph of stack-number $k$ has at most $(k+1)n-3k$ edges for $n\geq3$~\cite{NesetrilEtAl2012,BernhartKainen1979}.
Graphs of treewidth at most $k$ are $k$-degenerate, as are their minors, so $\mad(G)\leq2k$ and $\nabla_1(G)\leq k$.

For $K_k$-minor-free graphs with $3\leq k\leq9$, results of Dirac~\cite{Dirac1964} (for $3\leq k\leq5$), Mader~\cite{Mader1968} (for $k=6,7$), J\o rgensen~\cite{Jorgensen1994} (for $k=8$), and Song--Thomas~\cite{SongThomas2006} (for $k=9$) give $\mad(G)\leq2k-4$. See~\cite[Theorems~1.1 to 1.3]{SongThomas2006}.
Since a minor of a $K_k$-minor-free graph is also $K_k$-minor-free, this gives $\nabla_1(G)\leq k-2$.

Apply \cref{cor:mad} to get the results.

To prove the tightness up to $\eps$, use the graphs $P_n^{2k}$, $P_n^{k+1}$, $P_n^k$, and $P_n^{k-2}$, respectively.
The first graph has $k$ queues: group the edges of lengths $2i-1$ and $2i$ in its natural order, and the second has $k$ stacks~\cite{SwaminathanEtAl1995}.
One can easily provide a path-decomposition of $P_n^r$ of width $r$. This handles the last two classes.
\end{proof}

\section{Open questions}

We are very close to fully understanding the interaction between \MinDS and the \LOCAL model. In this section, we give a list of open questions that are left open by our work.

\paragraph{The fractional optimum.}
\begin{question}\label{q:fractional}
    Can the $(d+1+\eps)\MDS(G)$ bound in \cref{cor:mad} be strengthened to $(d+1+\eps)\MDS_f(G)$ in the bounded $\nabla_1$ case?
\end{question}
To get this result, it would suffice to replace the construction of the set $C$ with a fractional version.
We would start with a minimum fractional dominating set and increase its total weight only a little to satisfy the hyperedge constraints.
Amiri and Siebertz get this exact result for the smaller case of proper minor-closed classes~\cite{AmiriSiebertz2026}.

\paragraph{Weaker sparsity assumptions.}
\begin{question}\label{q:weaker-density}
    Which structural parameter weaker than bounded $\nabla_1$ suffices to get a deterministic constant-round \LOCAL $(d+1+\eps)$-approximation on graphs with $\mad(G)\leq d$?
\end{question}
Note that one cannot simply drop the bound on $\nabla_1$, because of the lower bound result of Dory, Ghaffari, and Ilchi on graphs of arboricity two~\cite{DoryEtAl2024}.

\paragraph{The exact approximation value.}
The matching lower bounds cover even integers $d$. Can this be improved to other values of $d$?
\begin{question}\label{q:tightness}
    For every real $d\geq2$, is there a class of graphs $\cC$ with $\mad(\cC)\leq d$ and bounded $\nabla_1(\cC)$ on which no constant-round algorithm achieves a ratio below $d+1$?
\end{question}
Our algorithm's running time depends heavily on $\eps$. This leads to the following question.
\begin{question}\label{q:exact-ratio}
    Can the $\eps$ be removed from the $d+1+\eps$ approximation ratio in \cref{cor:mad}?
\end{question}
Proving these two would give truly tight results for graphs of bounded $\nabla_1$.
Note that the simpler question about the existence of a $7$-approximation on planar graphs remains open~\cite[Section~9]{AmiriSiebertz2026}.

\paragraph{Small messages.} Our current algorithm works in the \LOCAL model. Therefore, the following question is natural:
\begin{question}\label{q:congest}
    Does \cref{cor:mad} hold in the \CONGEST model?
\end{question}
Note that our algorithm cannot be extended directly, as it uses very large messages to gather neighborhoods and enumerate hyperedge constraints.
Even achieving an approximation ratio of $7+\eps$ in constant \CONGEST rounds on planar graphs is open, see Amiri and Siebertz~\cite[Section~9]{AmiriSiebertz2026}.

\section*{Acknowledgments}

We thank Yi-Jun Chang, with whom this project started, for his many helpful comments and suggestions on the presentation of this article. We also thank Marthe Bonamy and Cyril Gavoille, for the discussions that led to the conjecture of the $\mad(G)+1$ factor in 2021.

\section*{Declaration of AI use}

\begin{table}[H]
\centering
\begin{tabular}{@{}l@{\qquad}c@{}}
\textbf{Activity} & \textbf{AI used} \\
Asking the research question & No \\
Coming up with the approach & Yes \\
Proof development & Yes \\
Writing and exposition & No \\
Checking for bugs & No \\
Other supporting tasks & Yes \\
\end{tabular}
\caption{Summary table about our use of AI, following the recommendations of the Working Group Report \emph{AI and TCS: The Next Six Months}. For more information, see \url{https://simons.berkeley.edu/AITCSreport}.}
\end{table}

We used AI tools to assist in finding results, proofreading, \LaTeX{} typesetting, and suggesting better English phrasing. The textual content was fully written by the author.
The main theorem was proved by AI with little input.
The author modified the proof to extend it to bounded genus graphs: in particular, by allowing for an additive bound in the density, and then by bruteforcing on small graphs.
The conjecture of an approximation factor of $\mad(G)+1$ was formulated by humans and had already been discussed by Marthe Bonamy, Cyril Gavoille, and the author around 2021.
The human author has independently verified the correctness and originality of all results in the paper and takes full responsibility for its contents.

\bibliographystyle{alphaurl}
\begingroup
\raggedright
\bibliography{bibliography}
\endgroup

\end{document}